\documentclass[11pt,onecolumn]{article}

\usepackage{fullpage}

\usepackage{epsfig}
\usepackage{color}
\usepackage{graphicx}
\usepackage{amssymb,latexsym,amsmath}
\usepackage[latin1]{inputenc}
\usepackage{enumerate}
\usepackage{url}
\usepackage{pstricks}
\usepackage{amsfonts,epsfig,graphicx}
\usepackage{amsmath,amssymb,amsthm}

\usepackage{multirow}
\usepackage{psfrag}

\makeatletter \long\def\@makecaption#1#2{
        \vskip 0.8ex
        \setbox\@tempboxa\hbox{\small {\bf #1:} #2}
        \parindent 1.5em  
        \dimen0=\hsize
        \advance\dimen0 by -3em
        \ifdim \wd\@tempboxa >\dimen0
                \hbox to \hsize{
                        \parindent 0em
                        \hfil
                        \parbox{\dimen0}{\def\baselinestretch{0.96}\small
                                {\bf #1.} #2
                                }
                        \hfil}
        \else \hbox to \hsize{\hfil \box\@tempboxa \hfil}
        \fi
        }
\makeatother

\newtheorem{theorem}{Theorem}
\newtheorem{lemma}{Lemma}
\newtheorem{corollary}{Corollary}

\theoremstyle{plain}

\theoremstyle{definition}

\newcommand{\widgraph}[2]{\includegraphics[keepaspectratio,width=#1]{#2}}

\newcommand{\conditionA}{\ensuremath{\textbf{A}}}

\newcommand{\defn}{\ensuremath{: \!\, =}}
\newcommand{\Expt}{\ensuremath{\mathbb{E}}}
\newcommand{\Prob}{\ensuremath{\mathbb{P}}}
\newcommand{\mprob}{\ensuremath{\mathbb{P}}}
\newcommand{\norm}[1]{\ensuremath{\|#1\|_{L^2}}}
\newcommand{\opnorm}[1]{\ensuremath{|\!|\!| #1 |\!|\!|_2}}
\newcommand{\inprod}[2]{\ensuremath{\langle #1, \, #2 \rangle}_{L^2}}
\newcommand{\coneleq}{\ensuremath{\preceq}}
\newcommand{\project}[1]{\ensuremath{\Pi_{\ball}\big(#1\big)}}

\newcommand{\as}{\stackrel{\text{a.s.}}{\longrightarrow}}

\newcommand{\convhull}{\operatorname{convhull}}

\newcommand{\diameter}{\operatorname{diam}}

\newcommand{\order}{\ensuremath{\mathcal{O}}}

\newcommand{\real}{\mathbb{R}}

\newcommand{\Space}{\mathcal{X}}                     
\newcommand{\FunSpace}{\mathcal{M}}                  
\newcommand{\FunSpaceTwo}{\mathcal{M}^{\prime}}

\newcommand{\ball}{\mathcal{B}}                      

\newcommand{\dimn}{\ensuremath{r}}                   

\newcommand{\graph}{\mathcal{G}}

\newcommand{\edge}{\mathcal{E}}

\newcommand{\DirSet}{\mathcal{\vec{\edge}}}
\newcommand{\numnode}{\ensuremath{n}}
\newcommand{\degr}{\ensuremath{d}}

\newcommand{\diam}{\ensuremath{\ell}}

\newcommand{\rv}{\ensuremath{X}}
\newcommand{\rvar}{\ensuremath{\rv}}
\newcommand{\realize}{\ensuremath{x}}

\newcommand{\Dimn}{\ensuremath{|\DirSet|}}

\newcommand{\epot}{\ensuremath{\psi_{\snode\fnode}}}
\newcommand{\npot}{\ensuremath{\psi_{\snode}}}
\newcommand{\poten}{\ensuremath{\psi}}

\newcommand{\fnode}{\ensuremath{v}}
\newcommand{\snode}{\ensuremath{u}}
\newcommand{\tnode}{\ensuremath{w}}
\newcommand{\ftnode}{\ensuremath{s}}

\newcommand{\Neig}{\ensuremath{\mathcal{N}}}
\newcommand{\CompNeig}{\ensuremath{\tnode\in \Neig (\fnode)
    \setminus \{\snode\}}}
\newcommand{\CompNeigtwo}{\ensuremath{\ftnode \in \Neig (\fnode)
    \setminus \{\snode, \tnode\}}}

\newcommand{\compat}{\ensuremath{\psi}}

\newcommand{\rvy}{\ensuremath{Y}}

\newcommand{\Time}{\ensuremath{t}}
\newcommand{\newt}{\ensuremath{\tau}}
\newcommand{\step}{\ensuremath{\eta^{\Time}}}

\newcommand{\mes}{\ensuremath{m}}
\newcommand{\mesprod}{\ensuremath{M}}
\newcommand{\mesprodtwo}{\ensuremath{M}_{\snode\fnode}}

\newcommand{\mesast}{\ensuremath{\mes^\ast}}
\newcommand{\marg}{\ensuremath{\tau}}

\newcommand{\upfun}{\ensuremath{\mathcal{F}}}

\newcommand{\probdens}{\ensuremath{p}}

\newcommand{\margfun}{\ensuremath{\beta}}
\newcommand{\compatfun}{\ensuremath{\Gamma}}                             
\newcommand{\prjcompatfun}{\ensuremath{\widetilde{\Gamma}}}                

\newcommand{\aprerr}{\ensuremath{\Delta}}

\newcommand{\error}{\ensuremath{e}}

\newcommand{\measure}{\mu}
\newcommand{\basis}{\phi}
\newcommand{\const}{\ensuremath{C}}
\newcommand{\bound}{\ensuremath{B}}

\newcommand{\PlainDeriv}{\ensuremath{\mathcal{D}}}
\newcommand{\Deriv}{\ensuremath{\PlainDeriv_{\tnode}
    \:\probdens_{\LDE}}}

\newcommand{\arbfun}{\ensuremath{h_{\tnode\rightarrow\fnode}}}
\newcommand{\func}{\ensuremath{f}}
\newcommand{\funchat}{\ensuremath{\widehat{\func}}}

\newcommand{\matentry}{\ensuremath{\widetilde{L}_{\LowerDirEdge{\fnode}{\snode},\tnode}}}

\newcommand{\mescoef}{\ensuremath{a}}                               

\newcommand{\coeff}{\ensuremath{a}}

\newcommand{\upfuncoef}{\ensuremath{b}}                             
\newcommand{\upfuncoeftwo}[2]{\ensuremath{b_{\fnode\rightarrow\snode; #2}^{#1}}}
\newcommand{\compupfuncoef}{\upfuncoef_{\fnode\rightarrow\snode; j}^{\Time}}              

\newcommand{\numsampY}{\ensuremath{k}}                         

\newcommand{\PlainSamp}{\ensuremath{Y}}

\newcommand{\devia}{\ensuremath{\zeta}}                           

\newcommand{\Lips}{\ensuremath{L_{\LowerDirEdge{\fnode}{\snode}; w}}}

\newcommand{\term}{\ensuremath{T}}
\newcommand{\detterm}{\ensuremath{D_{\LowerDirEdge{\fnode}{\snode}}^{\Time+1}}}
\newcommand{\stocterm}{\ensuremath{S_{\LowerDirEdge{\fnode}{\snode}}^{\Time+1}}}

\newcommand{\eigen}{\ensuremath{\lambda}}

\newcommand{\smooth}{\ensuremath{\alpha}}
\newcommand{\onevec}{\ensuremath{\vec{1}}}

\newcommand{\contfact}{\ensuremath{\gamma}}

\newcommand{\fframe}{\ensuremath{I}}
\newcommand{\sframe}{\ensuremath{I^{\prime}}}

\newcommand{\Graph}{\ensuremath{\mathcal{G}}}
\newcommand{\Vertex}{\ensuremath{\mathcal{V}}}
\newcommand{\Edge}{\ensuremath{\mathcal{E}}}

\newcommand{\DirEdge}[2]{\ensuremath{(#1 \rightarrow #2)}}
\newcommand{\LowerDirEdge}[2]{\ensuremath{ {#1} \rightarrow {#2}}}

\newcommand{\ALG}{SOSMP }
\newcommand{\stochastic}{stochastic }

\newcommand{\meshat}{\ensuremath{\widehat{\mes}}}

\newcommand{\Exs}{\ensuremath{\mathbb{E}}}

\newcommand{\coeffdirthree}[2]{\ensuremath{\coeff^{#1}_{\LowerDirEdge{\fnode}{\snode};
      #2}}}

\newcommand{\coefftilthree}[2]{\ensuremath{\widetilde{b}^{#1}_{\LowerDirEdge{\fnode}{\snode};#2}}}

\newcommand{\coeffup}[1]{\ensuremath{\coeff^{#1}_{\LowerDirEdge{\fnode}{\snode}}}}
\newcommand{\coefftil}[1]{\ensuremath{\widetilde{b}^{#1}_{\LowerDirEdge{\fnode}{\snode}}}}

\newcommand{\HACK}{\ensuremath{\gamma}}

\newcommand{\SUPERHACK}[1]{\ensuremath{\HACK_{\snode \fnode;  #1}}}

\newcommand{\coeffupdown}[2]{\ensuremath{\coeff^{#1}_{\LowerDirEdge{\fnode}{\snode};
      #2}}}

\newcommand{\fullcoeff}[1]{\ensuremath{{\coeff}^{#1}}}

\newcommand{\ItErr}[1]{\ensuremath{\Delta^{#1}_{\LowerDirEdge{\fnode}{\snode}}}}
\newcommand{\ItErrtwo}[2]{\ensuremath{\Delta^{#1}_{#2}}}
\newcommand{\ItErrNo}[1]{\ensuremath{\Delta^{#1}}}

\newcommand{\MASSERR}[1]{\ensuremath{\rho^2 \big(#1\big)}}

\newcommand{\GoodApp}[1]{\ensuremath{A^{#1}}}
\newcommand{\GoodApptwo}[2]{\ensuremath{A^{#1}_{#2}}}
\newcommand{\GoodAppDown}[1]{\ensuremath{A^{#1}_{\LowerDirEdge{\fnode}{\snode}}}}

\newcommand{\Pir}{\ensuremath{\Pi^\dimn}}

\newcommand{\mesup}[1]{\ensuremath{\mes^{#1}_{\LowerDirEdge{\fnode}{\snode}}}}

\newcommand{\compathat}{\ensuremath{\widehat{\Gamma}}}

\newcommand{\dimcrit}{\ensuremath{\dimn^*}}

\newcommand{\polydecay}{\ensuremath{\alpha}}

\newcommand{\myeig}{\ensuremath{\lambda}}

\newcommand{\coeffbp}{\ensuremath{\coeff^*}}
\newcommand{\coeffbptwo}[1]{\ensuremath{\coeffbp_{\LowerDirEdge{\fnode}{\snode}; #1}}}

\newcommand{\deviate}[1]{\ensuremath{\devia^{#1}_{\LowerDirEdge{\fnode}{\snode}}}}
\newcommand{\deviatetwo}[2]{\ensuremath{\devia^{#1}_{\LowerDirEdge{\fnode}{\snode}; #2}}}

\newcommand{\Field}{\ensuremath{\mathcal{G}}}

\newcommand{\NilMat}{\ensuremath{N}}

\newcommand{\phack}[2]{\ensuremath{[p_{\LowerDirEdge{\fnode}{\snode}}(#1)](#2)}}
\newcommand{\phackno}[1]{\ensuremath{p_{\LowerDirEdge{\fnode}{\snode}}(#1)}}

\newcommand{\vvec}{\ensuremath{\vec{v}}}

\newcommand{\NEWTERMONEDOWN}{\ensuremath{U^\Time_{\LowerDirEdge{\fnode}{\snode}}}}
\newcommand{\NEWTERMTWODOWN}{\ensuremath{V^\Time_{\LowerDirEdge{\fnode}{\snode}}}}

\newcommand{\LDE}{\ensuremath{{\LowerDirEdge{\fnode}{\snode}}}}
\newcommand{\neigh}{\ensuremath{\mathcal{N}}}
\newcommand{\AVEERR}[1]{\ensuremath{\widebar{\rho}^2(#1)}}
\newcommand{\RHOSQMAX}[1]{\ensuremath{\rho^2_{\tiny{\mbox{max}}}(#1)}}
\newcommand{\RHOBARSQMAX}[1]{\ensuremath{\widebar{\rho}^2_{\tiny{\mbox{max}}}(#1)}}
\newcommand{\SUMBOUND}{\ensuremath{\widebar{B}}}

\newcommand{\NumOp}{\ensuremath{T}}
\newcommand{\NumOpPoly}{\ensuremath{\NumOp_{\operatorname{\tiny{poly}}}}}
\newcommand{\pdens}{\ensuremath{p}}

\newlength{\widebarargwidth}
\newlength{\widebarargheight}
\newlength{\widebarargdepth}
\DeclareRobustCommand{\widebar}[1]{%
  \settowidth{\widebarargwidth}{\ensuremath{#1}}%
  \settoheight{\widebarargheight}{\ensuremath{#1}}%
  \settodepth{\widebarargdepth}{\ensuremath{#1}}%
  \addtolength{\widebarargwidth}{-0.3\widebarargheight}%
  \addtolength{\widebarargwidth}{-0.3\widebarargdepth}%
  \makebox[0pt][l]{\hspace{0.3\widebarargheight}%
    \hspace{0.3\widebarargdepth}%
    \addtolength{\widebarargheight}{0.3ex}%
    \rule[\widebarargheight]{0.95\widebarargwidth}{0.1ex}}%
  {#1}}

\begin{document}

\begin{center}

{\bf{\Large{Belief Propagation for Continuous State Spaces:
      \\ Stochastic Message-Passing with Quantitative Guarantees}}}

\vspace*{.2in}

{\large{
\begin{tabular}{ccc}
Nima Noorshams$^1$ & & Martin J. Wainwright$^{1,2}$ \\
{\texttt{nshams@eecs.berkeley.edu}} & & {\texttt{wainwrig@eecs.berkeley.edu}}
\end{tabular}
}}

\vspace*{.1in}

\begin{center}
Department of Statistics$^2$ and \\
Department of Electrical Engineering $\&$ Computer Science$^{1}$\\
University of California Berkeley 
\end{center}



\vspace*{.1in}

December 2012

\end{center}


\begin{abstract}

The sum-product or belief propagation (BP) algorithm is a widely used
message-passing technique for computing approximate marginals in
graphical models.  We introduce a new technique, called stochastic
orthogonal series message-passing (SOSMP), for computing the BP fixed
point in models with continuous random variables.  It is based on a
deterministic approximation of the messages via orthogonal series
expansion, and a stochastic approximation via Monte Carlo estimates of
the integral updates of the basis coefficients.  We prove that the
\ALG iterates converge to a $\delta$-neighborhood of the unique BP
fixed point for any tree-structured graph, and for any graphs with
cycles in which the BP updates satisfy a contractivity condition.  In
addition, we demonstrate how to choose the number of basis
coefficients as a function of the desired approximation accuracy
$\delta$ and smoothness of the compatiblity functions.  We illustrate
our theory with both simulated examples and in application to optical
flow estimation.

\end{abstract}



\section{Introduction}
\label{SecIntroduction}

Graphical models provide a parsimonious yet flexible framework for
describing probabilistic dependencies among large numbers of random
variables.  They have proven useful in a variety of application
domains, including computational biology, computer vision and image
processing, data compression, and natural language processing, among
others.  In all of these applications, a central computational
challenge is the \emph{marginalization problem}, by which we mean the
problem of computing marginal distributions over some subset of the
variables.  Naively approached, such marginalization problems become
intractable for all but toy problems, since they entail performing
summation or integration over high-dimensional spaces.  The
sum-product algorithm, also known as belief propagation (BP), is a
form of dynamic programming that can be used to compute exact
marginals much more efficiently for graphical models without cycles,
known as trees.  It is an iterative algorithm in which nodes in the
graph perform a local summation/integration operation, and then relay
results to their neighbors in the form of messages.  Although it is
guaranteed to be exact on trees, it is also commonly applied to graphs
with cycles, in which context it is often known as loopy BP.  For more
details on graphical models and belief propagation, we refer the
readers to the papers~\cite{KasEtal01, WaiJorBook08, AjiMce00,
  Loeliger04,Yedidia05}.

In many applications of graphical models, we encounter random
variables that take on continuous values (as opposed to discrete).
For instance, in computer vision, the problem of optical flow
calculation is most readily formulated in terms of estimating a vector
field in $\real^2$.  Other applications involving continuous random
variables include tracking problems in sensor networks, vehicle
localization, image geotagging, and protein folding in computational
biology.  With certain exceptions (such as multivariate Gaussian
problems), the marginalization problem is very challenging for
continuous random variables: in particular, the messages correspond to
functions, so that they are expensive to compute and transmit, in
which case belief propagation may be limited to small-scale problems.
Motivated by this challenge, researchers have proposed different
techniques to reduce complexity of BP in different
applications~\cite{AruEtal02,SuddEtal03,DouFreGor,IhlMcA09,
  CouShe07,IsaEtal09, SongEtal11, NooWai12a}.  For instance, various
types of quantization schemes~\cite{CouShe07,IsaEtal09} have used to
reduce the effective state space and consequently the complexity. In
another line of work, researchers have proposed stochastic methods
inspired by particle
filtering~\cite{AruEtal02,SuddEtal03,DouFreGor,IhlMcA09}. These
techniques are typically based on approximating the messages as
weighted particles~\cite{DouFreGor, IhlMcA09}, or mixture of
Gaussians~\cite{SuddEtal03}. Other researchers~\cite{SongEtal11} have
proposed the use of kernel methods to simultaneously estimate
parameters and compute approximate marginals in a simultaneous manner.

In this paper, we present a low-complexity alternative to belief
propagation with continous variables.  Our method, which we refer to
as stochastic orthogonal series message-passing (SOSMP), is applicable
to general graphical models, and is equipped with various theoretical
guarantees.  As suggested by its name, the algorithm is based on
combining two ingredients: orthogonal series approximation of the
messages, and the use of stochastic updates for efficiency. In this
way, the \ALG updates lead to a randomized algorithm with substantial
reductions in communication and computational complexity.  Our main
contributions are to analyze the convergence properties of the \ALG
algorithm, and to provide rigorous bounds on the overall error as a
function of the associated computational complexity.  In particular,
for tree-structured graphs, we estabish almost sure convergence, and
provide an explicit inverse polynomial convergence rate
(Theorem~\ref{ThmTree}).  For loopy graphical models on which the
usual BP updates are contractive, we also establish similar
convergence rates (Theorem~\ref{ThmGeneral}).  Our general theory
provides quantitative upper bounds on the number of iterations
required to compute a $\delta$-accurate approximation to the BP
message fixed point, as we illustrate in the case of kernel-based
potential functions (Theorem~\ref{ThmKernel}).

The reminder of the paper is organized as follows. We begin in
Section~\ref{SecBackground}, with the necessary background on the
graphical models as well as the belief propagation algorithm.
Section~\ref{SecAlgorithm} is devoted to a precise description of \ALG
algorithm.  In Section~\ref{SecTheory}, we state our main theoretical
results and develop some of their corollaries.  In order to
demonstrate the algorithm's effectiveness and confirm theoretical
predictions, we provide some experimental results, on both synthetic
and real data, in Section~\ref{SecSimulations}.  In
Section~\ref{SecProof}, we provide the proofs of our main results,
with some of the technical aspects deferred to the appendices.


\section{Background}
\label{SecBackground}

We begin by providing some background on graphical models and the
belief propagation (or sum-product) algorithm.

\subsection{Undirected graphical models}
\label{SecGraphicalModel}

Consider an undirected graph $\graph = (\Vertex, \Edge)$, consisting
of a collection of nodes or vertices \mbox{$\Vertex = \{1, 2, \ldots,
  \numnode \}$,} along with a collection of edges $\Edge \subset
\Vertex \times \Vertex$.  An edge is an undirected pair $(\snode,
\fnode)$, and self-edges are forbidden (meaning that $(\snode, \snode)
\notin \Edge$ for all $\snode \in \Vertex$).  For each $\snode \in
\Vertex$, let $\rv_\snode$ be a random variable taking values in a
space $\Space_\snode$.  An undirected graphical model, also known as a
Markov random field, defines a family of joint probability
distributions over the random vector $\rv = \{ \rv_\snode, \; \snode
\in \Vertex \}$, in which each distribution must factorize in terms of
local potential functions associated with the cliques of the graph.
In this paper, we focus on the case of pairwise Markov random fields,
in which case the factorization is specified in terms of functions
associated with the nodes and edges of the graph.

More precisely, we consider probability densities $\pdens$ that are
absolutely continuous with respect to a given measure $\measure$,
typically the Lebesgue measure for the continuous random variables
considered here.  We say that $\pdens$ \emph{respects the graph
  structure} if it can be factorized in the form
\begin{align}
\label{EqnPairwiseFact}
\pdens(\realize_1, \realize_2, \ldots, \realize_\numnode) \; & \propto
\; \prod_{\snode \in \Vertex} \poten_\snode (\realize_\snode)
\prod_{(\snode, \fnode) \in \Edge} \epot (\realize_\snode,
\realize_\fnode).
\end{align}
Here $\poten_\snode: \Space_\snode \rightarrow (0,\infty)$ is the node
potential function, whereas $\poten_{\snode\fnode}: \Space_\snode
\times \Space_\fnode \rightarrow (0,\infty)$ denotes the edge
potential function.  A factorization of this
form~\eqref{EqnPairwiseFact} is also known as \emph{pairwise Markov
  random field}; see Figure~\ref{FigGraphicalModels} for a few
examples that are widely used in practice.  

In many applications, a central computational challenge is the
computation of the marginal distribution
\begin{align}
\label{EqnDefnSingleNode}
\pdens(\realize_\snode) \;\; & \defn \!
\underbrace{\int_{\Space}\ldots\int_{\Space}}_{(\numnode-1) \;
  \text{times}} \pdens(\realize_1, \realize_2, \ldots,
\realize_\numnode) \; \prod_{\fnode \in \Vertex \backslash \{\snode\}}
\measure(dx_\fnode )
\end{align}
at each node $\snode \in \Vertex$.  Naively approached, this problem
suffers from the curse of dimensionality, since it requires computing
a multi-dimensional integral over an $(\numnode-1)$-dimensional space.
For Markov random fields defined on trees (graphs without cycles),
part of this exponential explosion can be circumvented by the use of
the belief propagation or sum-product algorithm, to which we turn in
the following section.

\begin{figure}[h]
\begin{center}
\begin{tabular}{ccc}
\psfrag{*cedge*}{$\poten_{\snode \fnode}$}
\psfrag{*cvar1*}{$\poten_{\fnode}$}
\psfrag{*cvar2*}{$\qquad \quad \poten_{\snode}$}
\widgraph{.45\textwidth}{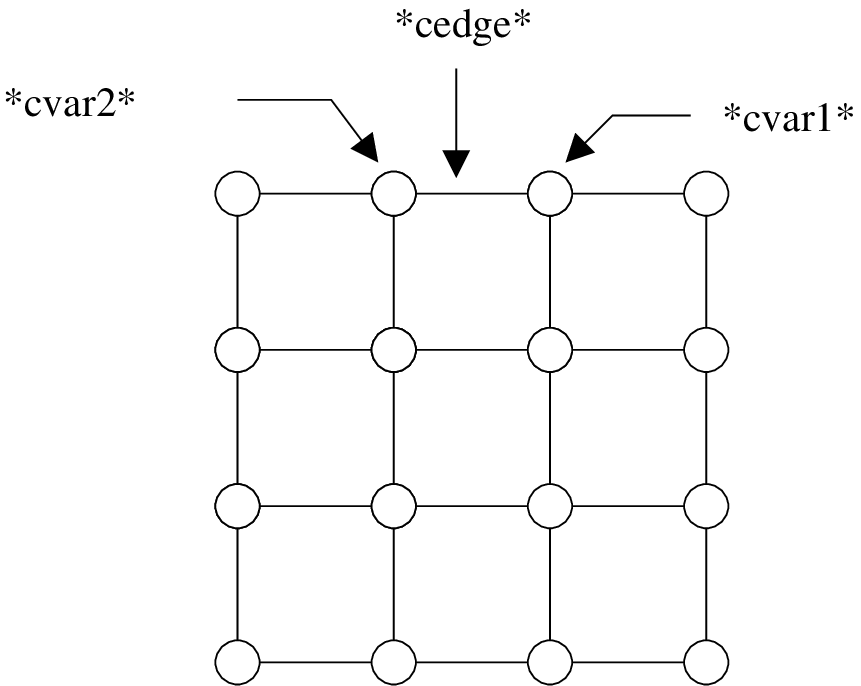} & & 
\psfrag{#cedge#}{$\poten_{\snode \fnode}$}
\psfrag{*cvar1*}{$\poten_{\fnode}$}
\psfrag{*cvar2*}{$\poten_{\snode}$}
\raisebox{.7in}{\widgraph{.45\textwidth}{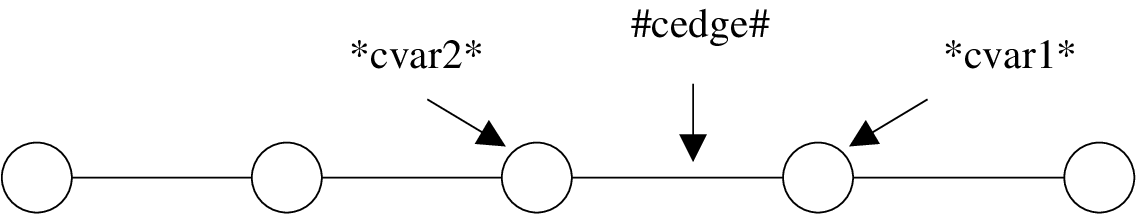}} \\
(a) & & (b)
\end{tabular}
\end{center}
\caption{Examples of pairwise Markov random fields. (a)
  Two-dimensional grid. (b) Markov chain model. Potential functions
  $\poten_\snode$ and $\poten_\fnode$ are associated with nodes
  $\snode$ and $\fnode$ respectively, whereas potential function
  $\poten_{\snode \fnode}$ is associated with edge $(\snode,
  \fnode)$.}
\label{FigGraphicalModels}
\end{figure}

Before proceeding, let us make a few comments about the relevance of
the marginals in applied problems.  In a typical application, one also
makes independent noisy observations $y_\snode$ of each hidden random
variable $\rvar_\snode$.  By Bayes' rule, the posterior distribution
of $\rvar$ given the observations $y = (y_1, \ldots, y_\numnode)$ then
takes the form
\begin{align}
\label{EqnPosterior}
\pdens_{X \mid Y} (\realize_1, \ldots, \realize_\numnode \, \mid y_1,
\ldots, y_\numnode) & \; \propto \; \prod_{\snode \in \Vertex}
\widetilde{\poten}_\snode (\realize_\snode; y_\snode) \prod_{(\snode,
  \fnode) \in \Edge} \epot (\realize_\snode, \realize_\fnode),
\end{align}
where we have introduced the convenient shorthand for the modified
node-wise potential functions
\mbox{$\widetilde{\poten}_\snode(\realize_\snode; y_\snode) \defn
  \pdens(y_\snode \mid \realize_\snode) \;
  \poten_\snode(\realize_\snode)$.}  Since the observation vector $y$
is fixed and known, any computational problem for the posterior
distribution~\eqref{EqnPosterior} can be reduced to an equivalent
problem for a pairwise Markov random field of the
form~\eqref{EqnPairwiseFact}, using the given definition of the
modified potential functions.  In addition, although our theory allows
for distinct state spaces $\Space_\snode$ at each node $\snode \in
\Vertex$, throughout the remainder of the paper, we suppress this
possible dependence so as to simplify exposition.


\subsection{Belief propagation}
\label{SecBP}

The belief propagation algorithm, also known as the sum-product
algorithm, is an iterative method based on message-passing updates for
computing either exact or approximate marginal distributions. For
trees (graphs without cycles), it is guaranteed to converge after a
finite number of iterations and yields the exact marginal
distributions, whereas for graphs with cycles, it yields only
approximations to the marginal distributions.  Nonetheless, this
``loopy'' form of belief propagation is widely used in practice. Here
we provide a very brief treatment sufficient for setting up the main
results and analysis of this paper, referring the reader to various
standard sources~\cite{KasEtal01,WaiJorBook08} for further background.

In order to define the message-passing updates, we require some
further notation.  For each node $\fnode \in \Vertex$, let
$\Neig(\fnode) \defn \{ \snode \in \Vertex \; \mid \; (\snode, \fnode)
\in \Edge \}$ be its set of neighbors, and we use
\mbox{$\DirSet(\fnode) \defn \{\DirEdge{\fnode}{\snode} \; \mid \;
  \snode \in \Neig (\fnode) \}$} to denote the set of all directed
edges emanating from $\fnode$.  We use $\DirSet \defn \cup_{\fnode\in
  \Vertex} \; \DirSet(\fnode)$ to denote the set of all directed edges
in the graph. Let $\FunSpace$ denote the set of all probability
densities (with respect to the base measure $\measure$) defined on the
space $\Space$---that is
\begin{align*}
\FunSpace& \; = \; \big \{ \mes: \Space \rightarrow [0, \infty) \, \mid
  \int_{\Space} \mes(\realize) \measure(d \realize) = 1 \big \}.
\end{align*}
The messages passed by the belief propagation algorithm are density
functions, taking values in the space $\FunSpace$.  More precisely, we
assign one message $\mes_{\LowerDirEdge{\fnode}{\snode}}
\in \FunSpace$ to every directed edge $\DirEdge{\fnode}{\snode} \in
\DirSet$, and we denote the collection of all messages by $\mes =
\{\mes_{\LowerDirEdge{\fnode}{\snode}}, \; \DirEdge{\fnode}{\snode}
\in \DirSet \}$.  Note that the full collection of messages $\mes$
takes values in the product space $\FunSpace^{|\DirSet|}$. \\

At an abstract level, the belief propagation algorithm generates a
sequence of message densities $\{\mes^t\}$ in the space
$\FunSpace^{|\DirSet|}$, where $t = 0,1,2 \ldots$ is the iteration
number.  The update of message $\mes^t$ to message $\mes^{t+1}$ can be
written in the form $\mes^{t+1} = \upfun(\mes^{t})$, where
\mbox{$\upfun: \FunSpace^{|\DirSet|}
  \rightarrow \FunSpace^{|\DirSet|}$} is a non-linear operator.  This
global operator is defined by the local update operators\footnote{ It
  is worth mentioning, and important for the computational efficiency
  of belief propagation, that $\mes_{\LowerDirEdge{\fnode}{\snode}}$
  is only a function of the messages
  $\mes_{\LowerDirEdge{\tnode}{\fnode}}$ for $\CompNeig$. Therefore,
  we have $\upfun_{\LowerDirEdge{\fnode}{\snode}}
  : \FunSpace^{\degr_{\fnode}-1} \to\FunSpace$, where $\degr_{\fnode}$
  is the degree of the node $\fnode$.  However, we suppress this local
  dependence so as to reduce notational clutter.}
$\upfun_{\LowerDirEdge{\fnode}{\snode}}: \FunSpace^{|\DirSet|}
\rightarrow \FunSpace$, one for each directed edge of the graph, such
that $\mes^{t+1}_{\LowerDirEdge{\fnode}{\snode}} =
\upfun_{\LowerDirEdge{\fnode}{\snode}}(\mes^t)$.   \\

More precisely, in terms of these local updates, the BP algorithm
operates as follows.  At each iteration $\Time = 0, 1, \ldots$, each
node $\fnode \in \Vertex$ performs the following steps:
\begin{itemize}
\item for each one of its neighbors $\snode \in \Neig(\fnode)$, it
  computes $\mes^{t+1}_{\LowerDirEdge{\fnode}{\snode}} =
  \upfun_{\LowerDirEdge{\fnode}{\snode}}(\mes^t)$.
\item it transmits message $\mes^{t+1}_{\LowerDirEdge{\fnode}{\snode}}$ to
neighbor $\snode \in \Neig(\fnode)$.
\end{itemize}
In more detail, the message update takes the form
\begin{align}
\label{EqnBPUpdateEntry}
\underbrace{ [\upfun_{\LowerDirEdge{\fnode}{\snode}} (\mes^{\Time})]
  (\cdot)}_{\mes^{\Time + 1}_{\LowerDirEdge{\fnode}{\snode}} (\cdot)}
\; & \defn \; \kappa \int_{\Space} \Big \{ \poten_{\snode \fnode}(\cdot,
\realize_\fnode) \; \poten_{\fnode} (\realize_\fnode) \; \!
\prod_{\CompNeig} \mes^{\Time}_{\LowerDirEdge{\tnode}{\fnode}}
(\realize_\fnode) \Big \} \, \measure(d \realize_\fnode),
\end{align}
where $\kappa$ is a normalization constant chosen to enforce the
normalization condition 
\begin{equation*}
\int_{\Space}
\mes^{\Time+1}_{\LowerDirEdge{\fnode}{\snode}}(\realize_\snode) \:
\measure(d \realize_\snode) \; = \; 1.
\end{equation*}
By concatenating the local updates~\eqref{EqnBPUpdateEntry}, we obtain
a global update operator \mbox{$\upfun: \FunSpace^{|\DirSet|}
  \to\FunSpace^{|\DirSet|}$,} as previously discussed.  The goal of
belief propagation message-passing is to obtain a \emph{fixed point},
meaning an element $\mesast \in \FunSpace^{|\DirSet|}$ such that
$\upfun (\mesast) = \mesast$.  Under mild conditions, it can be shown
that there always exists at least one fixed point, and for any
tree-structured graph, the fixed point is unique.

Given a fixed point $\mesast$, each node $\snode \in \Vertex$ computes
its marginal approximation $\marg^{\ast}_\snode \in \FunSpace$ by
combining the local potential function $\poten_\snode$ with a product
of all incoming messages as
\begin{align}
\label{EqnUpdateMarg}
\marg^{\ast}_\snode (\realize_\snode) \; & \propto \;
\poten_{\snode}(\realize_\snode) \!  \prod_{\fnode \in \Neig(\snode)}
\mesast_{\LowerDirEdge{\fnode}{\snode}} (\realize_\snode).
\end{align}
Figure~\ref{FigMessPass} provides a graphical representation of the
flow of the information in these local updates. For tree-structured
(cycle-free) graphs, it is known that BP
updates~\eqref{EqnBPUpdateEntry} converge to the unique fixed point in
a finite number of iterations~\cite{WaiJorBook08}. Moreover, the
quantity $\marg^\ast_\snode(x_\snode)$ is equal to the single-node
marginal, as previously defined~\eqref{EqnDefnSingleNode}. For general
graphs, uniqueness of the fixed point is no longer
guaranteed~\cite{WaiJorBook08}; however, the same message-passing
updates can be applied, and are known to be extremely effective for
computing approximate marginals in numerous applications.

\begin{figure}[h]
\begin{center}
\begin{tabular}{cc}
\psfrag{#1#}{$w$}
\psfrag{#2#}{$\snode$} 
\psfrag{#3#}{$\fnode$}
\psfrag{#4#}{$s$} 
\psfrag{#5#}{$t$} 
\psfrag{#m43#}{$\mes_{\LowerDirEdge{s}{\fnode}}$}
\psfrag{#m53#}{$\mes_{\LowerDirEdge{t}{\fnode}}$}
\psfrag{#m32#}{$\mes_{\LowerDirEdge{\fnode}{\snode}}$}
\widgraph{0.45\textwidth}{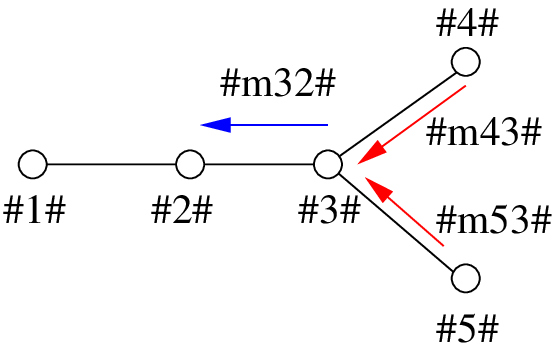} &
\psfrag{#2#}{$s$} 
\psfrag{#3#}{$\snode$}
\psfrag{#4#}{$t$} 
\psfrag{#5#}{$w$} 
\psfrag{#m43#}{$\mes_{\LowerDirEdge{t}{\snode}}$}
\psfrag{#m53#}{$\mes_{\LowerDirEdge{w}{\snode}}$} 
\psfrag{#m32#}{$\mes_{\LowerDirEdge{s}{\snode}}$}
\widgraph{0.35\textwidth}{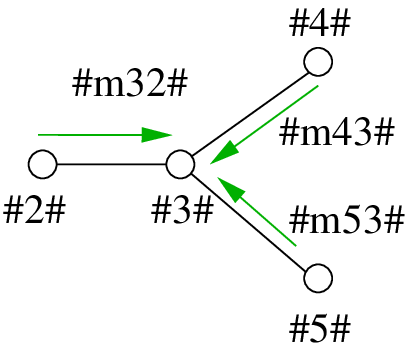} \\
(a) & (b)
\end{tabular}
\end{center}
\caption{Graphical representation of message-passing algorithms. (a)
  Node $\fnode$ transmits the message
  $\mes_{\LowerDirEdge{\fnode}{\snode}} =
  \upfun_{\LowerDirEdge{\fnode}{\snode}} (\mes)$, derived from
  equation~\eqref{EqnBPUpdateEntry}, to its neighbor $\snode$. (b)
  Upon receiving all the messages, node $\snode$ updates its marginal
  estimate according to~\eqref{EqnUpdateMarg}.}
\label{FigMessPass}
\end{figure}

Although the BP algorithm is considerably more efficient than the
brute force approach to marginalization, the message update
equation~\eqref{EqnBPUpdateEntry} still involves computing an integral
and transmitting a real-valued function (message).  With certain
exceptions (such as multivariate Gaussians), these continuous-valued
messages do not have finite representations, so that this approach is
computationally very expensive.  Although integrals can be computed by
numerical methods, the BP algorithm requires performing many such
integrals \emph{at each iteration}, which becomes very expensive in
practice.


\section{Description of the algorithm}
\label{SecAlgorithm}

We now turn to the description of the \ALG algorithm.  Before doing
so, we begin with some background on the main underlying ingredients:
orthogonal series expansion, and \stochastic message updates.

\subsection{Orthogonal series expansion}
\label{SecQSBP}

As described in the previous section, for continuous random variables,
each message is a density function in the space $\FunSpace \subset
L^2(\Space; \measure)$.  We measure distances in this space using the
usual $L^2$ norm $\|f - g\|_2^2 \defn \int_\Space (f(x) - g(x))^2
\measure(d x)$.  A standard way in which to approximate functions is
via orthogonal series expansion.  In particular, let
$\{\basis_j\}_{j=1}^{\infty}$ be an orthonormal basis of $L^2(\Space;
\measure)$, meaning a collection of functions such that
\begin{align}
\underbrace{\int_{\Space} \basis_i(x)\basis_j(x) \measure(dx)}_{\defn
  \inprod{\basis_i}{\basis_j}} \; & = \; \begin{cases} 1 & \mbox{when $i =
    j$} \\ 0 & \mbox{otherwise.}
\end{cases}
\end{align}
Any function $\func \in \FunSpace \subset L^2(\Space; \measure)$ then
has an expansion of the form $\func = \sum_{j=1}^\infty \coeff_j \phi_j$,
where $\coeff_j = \inprod{\func}{\phi_j}$ are the basis expansion
coefficients.

Of course, maintaining the infinite sequence of basis coefficients
$\{\coeff_j\}_{j=1}^\infty$ is also computationally intractable, so
that any practical algorithm will maintain only a finite number
$\dimn$ of basis coefficients.  For a given $\dimn$, we let
$\funchat_\dimn \propto \big[ \sum_{j=1}^\dimn \coeff_j \basis_j
\big]_+$ be the approximation based on the first $\dimn$ coefficients.
(Applying the operator $[t]_+ = \max \{0, t\}$ amounts to projecting
$\sum_{j=1}^\dimn \coeff_j \phi_j$ onto the space of non-negative
functions, and we also normalize to ensure that it is a density
function.)  In using only $\dimn$ coefficients, we incur the
\emph{approximation error}
\begin{align}
\label{EqnRapprox}
\| \funchat_\dimn - \func \|_2^2 \; & \stackrel{(i)}{\leq} \; \| \sum_{j=1}^\dimn
\coeff_j \phi_j - \func \|_2^2 \; \stackrel{(ii)}{=} \; \sum_{j = \dimn
  +1}^\infty \coeff_j^2
\end{align}
where inequality (i) uses non-expansivity of the projection, and step
(ii) follows from Parseval's theorem.  Consequently, the approximation
error will depend both on 
\begin{itemize}
\item how many coefficients $\dimn$ that we retain, and
\item the decay rate of the expansion coefficients
  $\{\coeff_j\}_{j=1}^\infty$.
\end{itemize}

For future reference, it is worth noting that the local message
update~\eqref{EqnBPUpdateEntry} is defined in terms of an integral
operator of the form
\begin{align}
\label{EqnIntegralUpdate}
\func(\cdot) \; \mapsto \; \int_\Space \poten_{\snode \fnode}(\cdot,
\realize_\fnode) \: \func(\realize_\fnode)\:
\measure(d \realize_\fnode).
\end{align}
Consequently, whenever the edge potential function $\poten_{\snode
  \fnode}$ has desirable properties---such as differentiability and/or
higher order smoothness---then the messages also inherit these
properties.  With an appropriate choice of the basis
$\{\basis_j\}_{j=1}^\infty$, such properties translate into decay
conditions on the basis coefficients $\{\coeff_j\}_{j=1}^\infty$.  For
instance, for $\alpha$-times differentiable functions expanded into
the Fourier basis, the Riemann-Lebesgue lemma guarantees that the
coefficients $\coeff_j$ decay faster than $(1/j)^{2 \alpha}$.  We
develop these ideas at greater length in the sequel.

\subsection{Stochastic message updates}

In order to reduce the approximation error~\eqref{EqnRapprox}, the
number of coefficients $\dimn$ needs to be increased (as a function of
the ultimate desired error $\delta$).  Since increases in $\dimn$ lead
to increases in computational complexity, we need to
develop effective reduced-complexity methods.  In this section, we
describe (at a high-level) how this can be done via a stochastic
version of the BP message-passing updates.

We begin by observing that message update~\eqref{EqnBPUpdateEntry},
following some appropriate normalization, can be cast as an
expectation operation.  This equivalence is essential, because it
allows us to obtain unbiased approximations of the message update
using \stochastic techniques. In particular, let us define the
\emph{normalized compatibility function}
\begin{align}
\label{EqnNormCompat}
\compatfun_{\snode \fnode} (\cdot, \realize_\fnode ) \defn \;
\epot(\cdot, \realize_\fnode)
\frac{\poten_\fnode(\realize_\fnode)}{\margfun_{\snode
    \fnode}(\realize_\fnode)}, \quad \mbox{where} \quad
\margfun_{\snode\fnode} (\realize_\fnode) \defn \; \poten_{\fnode}
(\realize_\fnode) \int_{\Space} \epot (\realize_\snode,
\realize_\fnode) \: \measure(d \realize_\snode).
\end{align}
By construction, for each $\realize_\fnode$, we have $\int_{\Space}
\compatfun_{\snode \fnode}(\realize_\snode, \realize_\fnode)
\measure(d \realize_\snode) = 1$.

\begin{lemma}
\label{LemCuteObservation}
Given an input collection of messages $\mes$, let $Y$ be a random
variable with density proportional to
\begin{align}
\label{EqnPdens}
\phack{\mes}{y} \; & \propto \; \margfun_{\snode \fnode}(y) \prod
\limits_{\CompNeig} \mes_{\LowerDirEdge{\tnode}{\fnode}}(y).
\end{align}
Then the message update equation~\eqref{EqnBPUpdateEntry} can be
written as
\begin{align}
[\upfun_{\LowerDirEdge{\fnode}{\snode}} (\mes)](\cdot) \; & = \;
\Exs_Y \big[ \compatfun_{\snode \fnode}(\cdot, Y) \big].
\end{align}
\end{lemma}
\begin{proof}
Let us introduce the convenient shorthand $\mesprod (y) = \prod
\limits_{\CompNeig} \mes_{\LowerDirEdge{\tnode}{\fnode}}(y)$.  By
  definition~\eqref{EqnBPUpdateEntry} of the message update, we have
\begin{align*}
[\upfun_{\LowerDirEdge{\fnode}{\snode}} (\mes)] (\cdot) \; = \; \frac
{\int_{\Space} \big(\epot(\cdot,y) \: \npot(y) \, \mesprod(y) \measure
  (dy)} {\int_{\Space}\int_{\Space} \big(\epot(x,y) \: \npot(y)
  \mesprod (y) \big) \: \measure (dy) \:\measure(dx)}.
\end{align*}
Since the integrand is positive, by Fubini's theorem~\cite{Durrett95},
we can exchange the order of integrals in the denominator. Doing so
and simplifying the expression yields 
\begin{align}
\label{EqnDefProbDens}
[\upfun_{\LowerDirEdge{\fnode}{\snode}} (\mes)] (\cdot) \; = \;
\int_{\Space} \underbrace{\frac{\epot (\cdot, y)}{\int_{\Space} \epot
    (x,y) \measure (dx)}}_{\compatfun_{\snode \fnode}(\cdot ,y)} \;
\underbrace{\frac{\margfun_{\snode\fnode}(y) \mesprod (y)}
  {\int_{\Space}\margfun_{\snode\fnode}(z) \mesprod (z)
    \measure(dz)}}_{\phack{\mes}{y}} \; \measure(dy),
\end{align}
which establishes the claim.
\end{proof}

Based on Lemma~\ref{LemCuteObservation}, we can obtain a \stochastic
approximation to the message update by drawing $\numsampY$
i.i.d. samples $Y_i$ from the density~\eqref{EqnPdens}, and then
computing $\sum_{i=1}^\numsampY \compatfun_{\snode \fnode}(\cdot, Y_i)
\: / \: \numsampY$.  Given the non-negativity and chosen normalization
of $\compatfun_{\snode \fnode}$, note that this estimate belongs to
$\FunSpace$ by construction.  Moreover, it is an unbiased estimate of
the correctly updated message, which plays an important role in our
analysis.

\subsection{Precise description of the algorithm}

The \ALG algorithm involves a combination of the orthogonal series
expansion techniques and \stochastic methods previously described.
Any particular version of the algorithm is specified by the choice of
basis $\{\basis_j\}_{j=1}^\infty$ and two positive integers: the
number of coefficients $\dimn$ that are maintained, and the number of
samples $\numsampY$ used in the \stochastic update.  Prior to running
the algorithm, for each directed edge $\DirEdge{\fnode}{\snode}$, we
pre-compute the inner products
\begin{align}
\label{EqnHACK}
\SUPERHACK{j} (\realize_\fnode) \; & \defn \; \underbrace{\int_\Space
  \compatfun_{\snode \fnode}(\realize_\snode, \realize_\fnode)
  \basis_j(\realize_\snode) \, \measure(d
  \realize_\snode),}_{\inprod{\compatfun_{\snode \fnode}(\cdot, \;
    \realize_\fnode)}{\phi_j(\cdot)}} \qquad
\mbox{for $j = 1, \ldots, \dimn$.}
\end{align}
When $\poten_{\snode \fnode}$ is a symmetric and positive semidefinite
kernel function, these inner products have an explicit and simple
representation in terms of its Mercer eigendecomposition (see
Section~\ref{SecKernel}).  In the general setting, these $\dimn$ inner
products can be computed via standard numerical integration
techniques.  Note that this is a fixed (one-time) cost prior to
running the algorithm.

\begin{figure}[ht!]
\framebox[.99\textwidth]{
\parbox{0.96\textwidth}{
{\bf{\ALG algorithm for marginalization:}}
\begin{enumerate}
\item At time $t = 0$, initialize the message coefficients
\begin{align*}
\coeffdirthree{0}{j} & =  1/\dimn \quad \mbox{for all $j = 1, \ldots,
  \dimn$, and $\DirEdge{\fnode}{\snode} \in \DirSet$.}
\end{align*}

\item For iterations $t = 0, 1, 2, \ldots$, and for each directed edge
  $\DirEdge{\fnode}{\snode}$
\begin{enumerate}
\item[(a)] Form the projected message approximation
  \mbox{$\meshat^\Time_{\LowerDirEdge{\tnode}{\fnode}}(\cdot) = \Big[
      \sum_{j=1}^\dimn
      \coeff^{\Time}_{\LowerDirEdge{\tnode}{\fnode};j} \basis_j(\cdot)
      \Big]_+$},
for all $\CompNeig$.
\item[(b)] Draw $\numsampY$ i.i.d. samples $\PlainSamp_i$ from the
  probability density proportional to
\begin{align}
\label{EqnProbDensity}
 \margfun_{\snode\fnode}(y) \; \prod_{\CompNeig} \!\!
 \meshat_{\LowerDirEdge{\tnode}{\fnode}}^{\Time}(y),
\end{align}
where $\margfun_{\snode \fnode}$ was previously defined in
equation~\eqref{EqnNormCompat}.
\item[(c)] Use the samples $\{\PlainSamp_1, \ldots,
  \PlainSamp_\numsampY \}$ from step (b) to compute
\begin{align}
\label{EqnUpdateInov}
\coefftilthree{\Time+1}{j} \; & \defn \; \frac{1}{\numsampY} \:
\sum_{i=1}^{\numsampY} \: \SUPERHACK{j}(\PlainSamp_i) \qquad \mbox{for $j
  = 1, 2, \ldots, \dimn$,}
\end{align}
where the function $\SUPERHACK{j}$ is defined in
equation~\eqref{EqnHACK}.

\item For step size $\step = 1/(\Time+1)$, update the
  $\dimn$-dimensional message coefficient vectors $\coeffup{t} \mapsto
  \coeffup{t+1}$ via
\begin{align}
\label{EqnUpdateCoef}
\coeffup{\Time+1} \; & = \; (1 - \step) \: \coeffup{\Time} \, + \, \step \:
\coefftil{\Time+1}.
\end{align}
\end{enumerate}
\end{enumerate}
} 
} 
\caption{The $\ALG$ algorithm for continuous state spaces.}
\label{FigAlg}
\end{figure}

At each iteration $\Time = 0, 1, 2, \ldots$, the algorithm maintains
an $\dimn$-dimensional vector of basis expansion coefficients
\begin{align*}
\coeffup{\Time} = (\coeffupdown{\Time}{1}, \ldots,
\coeffupdown{\Time}{\dimn}) \in \real^\dimn, \quad \mbox{on directed
  edge $\DirEdge{\fnode}{\snode} \in \DirSet$.}
\end{align*}
This vector should be understood as defining the current message
approximation $\mes^\Time_{\LowerDirEdge{\fnode}{\snode}}$ on edge
$\DirEdge{\fnode}{\snode}$ via the expansion
\begin{align}
\label{EqnCoeff2Mess}
\mes^{\Time}_{\LowerDirEdge{\fnode}{\snode}}(\cdot) \; \defn \;
\sum_{j=1}^\dimn \coeffupdown{\Time}{j} \; \basis_j(\cdot)
\end{align}
We use $\fullcoeff{\Time} = \big \{ \coeffup{\Time},
\DirEdge{\fnode}{\snode} \in \DirSet \big \}$ to denote the full set
of $\dimn \, |\DirSet|
$ coefficients that are maintained by the
algorithm at iteration $\Time$.  With this notation, the algorithm
consists of a sequence of steps, detailed in Figure~\ref{FigAlg}, that
perform the update $\fullcoeff{\Time} \mapsto \fullcoeff{\Time+1}$,
and hence implicitly the update $\mes^{\Time} \mapsto \mes^{\Time+1}$.

As can be seen by inspection of the steps in Figure~\ref{FigAlg}, each
iteration requires $\order(\dimn \numsampY)$ floating point operations
per directed edge, which yields a total of $\order(\dimn \numsampY \,
|\DirSet|)$ operations per iteration.


\section{Theoretical guarantees}
\label{SecTheory}

We now turn to a theoretical analysis of the \ALG algorithm, and
guarantees relative to the fixed points of the true BP algorithm.  For
any tree-structured graph, the BP algorithm is guaranteed to have a
unique message fixed point $\mes^* = \{
\mes^*_{\LowerDirEdge{\fnode}{\snode}}, \; \DirEdge{\fnode}{\snode}
\in \DirSet\}$.  For graphs with cycles, uniqueness is no longer
guaranteed, which would make it difficult to compare with the \ALG
algorithm.  Accordingly, in our analysis of the loopy graph, we make a
natural contractivity assumption, which guarantees uniqueness of the
fixed point $\mes^*$. \\

The \ALG algorithm generates a random sequence
$\{\coeff^\Time\}_{\Time=0}^\infty$, which define message
approximations $\{\mes^\Time\}_{\Time=0}^\infty$ via the
expansion~\eqref{EqnCoeff2Mess}.  Of interest to us are the following
questions:
\begin{itemize}
\item under what conditions do the message iterates approach a
  neighborhood of the BP fixed point $\mes^*$ as $\Time \rightarrow
  +\infty$?
\item when such convergence takes place, how fast is it?
\end{itemize}

\noindent In order to address these questions, we separate the error
in our analysis into two terms: algorithmic error and approximation
error.  For a given $\dimn$, let $\Pir$ denote the projection operator
onto the span of $\{\basis_1, \ldots, \basis_\dimn\}$.  In detail,
given a function $f$ represented in terms of the infinite series
expansion $f = \sum_{j=1}^\infty \coeff_j \basis_j$, we have
\begin{align*}
\Pir(f) \; & \defn \; \sum_{j=1}^\dimn \coeff_j \basis_j.
\end{align*}
 For each directed edge $\DirEdge{\fnode}{\snode} \in \DirSet$, define
 the functional error
\begin{align}\label{EqnDefEstErr}
 \ItErr{\Time} \; \defn \; \mesup{\Time} - \Pir(\mesup{*})
\end{align}
 between the message approximation at time $\Time$, and the BP fixed
 point projected onto the first $\dimn$ basis functions. Moreover,
 define the approximation error at the BP fixed point as
\begin{align}\label{EqnDefApproxErr}
\GoodAppDown{\dimn} \; \defn \; \mes^*_\LDE - \Pir(\mes^*_\LDE).
\end{align}
Since $\ItErr{\Time}$ belongs to the span of the first $\dimn$ basis
functions, the Pythagorean theorem implies that the overall error can
be decomposed as
\begin{align}
\| \mesup{\Time} - \mesup{*} \|_{L^2}^2 \quad \; & = \; \underbrace{\|
  \ItErr{\Time}\|_{L^2}^2}_{\mbox{Estimation error}} \quad + \;
\underbrace{\norm{\GoodAppDown{\dimn}}^2}_{\mbox{Approximation error}}
\end{align}
Note that the approximation error term is independent of the iteration
number $\Time$, and can only be reduced by increasing the number
$\dimn$ of coefficients used in the series expansion.  Our analysis of
the estimation error is based on controlling the
$|\DirSet|$-dimensional error vector
\begin{align}
\label{EqnDefnMassErr}
\MASSERR{\ItErrNo{\Time}} \; & \defn \; \big \{ \|\ItErr{\Time} \|_{L^2}^2,
\; \DirEdge{\fnode}{\snode} \in \DirSet \big \} \; \in \; \real^{|\DirSet|},
\end{align}
and in particular showing that it decreases as $\order(1/\Time)$ up to
a lower floor imposed by the approximation error. In order to analyze
the approximation error, we introduce the $\Dimn$-dimensional vector
of approximation errors
\begin{align}\label{EqnDefnMassErrApp}
\MASSERR{\GoodApp{\dimn}} & \defn \big \{ \|\GoodAppDown{\dimn }
\|_{L^2}^2, \; \DirEdge{\fnode}{\snode} \in \DirSet \big \} \in
\real^{|\DirSet|}.
\end{align}
By increasing $\dimn$, we can reduce this approximation error term,
but at the same time, we increase the computational complexity of each
update.  In Section~\ref{SecKernel}, we discuss how to choose $\dimn$
so as to trade-off the estimation and approximation errors with
computational complexity.


\subsection{Bounds for tree-structured graphs}

With this set-up, we now turn to bounds for tree-structured graphs.
Our analysis of the tree-structured case controls the vector of errors
$\MASSERR{\ItErrNo{\Time}}$ using a nilpotent matrix $\NilMat \in
\real^{\Dimn \times \Dimn}$ determined by the tree
structure~\cite{NooWai12a}. Recall that a matrix $\NilMat$ is
nilpotent with order $\ell$ if $\NilMat^\ell = 0$ and
$\NilMat^{\ell-1} \neq 0$ for some $\ell$. As illustrated in
Figure~\ref{FigNilpotent}, the rows and columns of $\NilMat$ are
indexed by directed edges. For the row indexed by
$\DirEdge{\fnode}{\snode}$, there can be non-zero entries only for
edges in the set $\{ \DirEdge{w}{\fnode}, \; w \in \Neig(\fnode)
\backslash \{\snode\} \}$. These directed edges are precisely those
that pass messages relevant in updating the message from $\fnode$ to
$\snode$, so that $\NilMat$ tracks the propagation of message
information in the graph. As shown in our previous work (see Lemma 1
in the paper~\cite{NooWai12a}), the matrix $\NilMat$ with such
structure is nilpotent with degree at most the diameter of the
tree. (In a tree, there is always a unique edge-disjoint path between
any pair of nodes; the diameter of the tree is the length of the
longest of these paths.)

\begin{figure}[h]
\begin{center}
\begin{tabular}{ccc}
\raisebox{.4in}{
\begin{psfrags}
\psfrag{*1*}{$1$}
\psfrag{*2*}{$2$}
\psfrag{*3*}{$3$}
\psfrag{*4*}{$4$}
\widgraph{.26\textwidth}{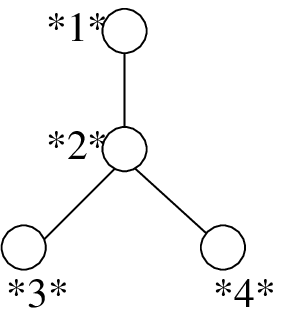}
\end{psfrags}
}
 & \quad & 
\begin{psfrags}
\psfrag{*m12*}{$\mes_{\LowerDirEdge{1}{2}}$}
\psfrag{*m21*}{$\mes_{\LowerDirEdge{2}{1}}$}
\psfrag{*m32*}{$\mes_{\LowerDirEdge{3}{2}}$}
\psfrag{*m23*}{$\mes_{\LowerDirEdge{2}{3}}$}
\psfrag{*m24*}{$\mes_{\LowerDirEdge{2}{4}}$}
\psfrag{*m42*}{$\mes_{\LowerDirEdge{4}{2}}$}
\widgraph{.5\textwidth}{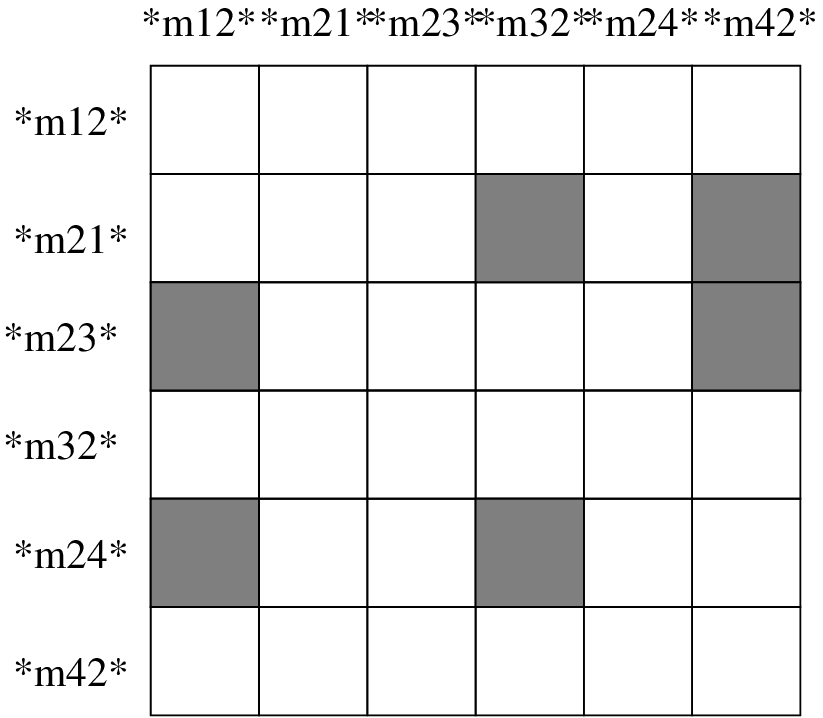}
\end{psfrags} \\
(a) & \quad & (b)
\end{tabular}
\caption{(a) A simple tree with $|\edge| = 3$ edges and hence
  $|\DirSet| = 6$ directed edges. (b) Structure of nilpotent matrix
  $\NilMat \in \real^{|\DirSet| \times |\DirSet|}$ defined by the
  graph\ in (a).  Rows and columns of the matrix are indexed by
  directed edges $\DirEdge{\fnode}{\snode} \in \DirSet$; for the row
  indexed by $\DirEdge{\fnode}{\snode}$, there can be non-zero entries
  only for edges in the set $\{ \DirEdge{w}{\fnode}, \; w \in
  \Neig(\fnode) \backslash \{\snode\} \}$.}
\label{FigNilpotent}
\end{center}
\end{figure}


Moreover, our results on tree-structured graphs impose one condition
on the vector of approximation errors $\GoodApp{\dimn}$, namely that
\begin{align}\label{EqnAnnoy}
\inf_{y\in\Space}\Pir\big(\compatfun_{\snode\fnode}(x, y) \big) \; >
\; 0, \quad \text{and} \quad |\GoodAppDown{\dimn}(x)| \; \le \;
\frac{1}{2} \: \inf_{y \in \Space} \Pir
\big(\compatfun_{\snode\fnode}(x, y) \big)
\end{align}
for all $x\in\Space$ and all directed edges $\DirEdge{\fnode}{\snode}
\in \DirSet$.  This condition ensures that the $L^2$-norm of the
approximation error is not too large relative to the compatibility
functions. Since
$\sup_{x, y\in\Space}|\Pir\big(\compatfun_{\snode\fnode}(x, y) \big) -
\compatfun_{\snode\fnode}(x, y)| \to 0$ and
$\sup_{x\in\Space}|\GoodAppDown{\dimn}(x)| \to 0$ as $\dimn \to
+\infty$, assuming that the compatibility functions are uniformly
bounded away from zero, condition~\eqref{EqnAnnoy} will hold once the
number of basis expansion coefficients $\dimn$ is sufficiently
large. Finally, our bounds involve the constants
\begin{align}
\label{EqnDefnBound}
\bound_j \; & \defn \; \max \limits_{\DirEdge{\fnode}{\snode} \in
  \DirSet} \: \sup \limits_{y \in \Space} \:
\inprod{\compatfun_{\snode\fnode}(\cdot, y)}{\basis_j}.
\end{align}
With this set-up, we have the following guarantees:

\begin{theorem}
\label{ThmTree} 
Suppose that $\Space$ is closed and bounded, the node and edge
potential functions are continuous, and that
condition~\eqref{EqnAnnoy} holds.  Then for any tree-structured model,
the sequence of messages $\{\mes^{\Time}\}_{\Time = 0}^{\infty}$
generated by the \ALG algorithm have the following properties:
\begin{enumerate}
\item[(a)] There is a nilpotent matrix $\NilMat \in \real^{|\DirSet|
  \times |\DirSet|}$ such that the error vector
  $\MASSERR{\ItErrNo{\Time}}$ converges almost surely to the set
\begin{align}
\ball & \defn \big \{ e \in \real^{|\DirSet|} \; \mid \; |e| \coneleq
\NilMat (I-\NilMat)^{-1} \MASSERR{\GoodApp{\dimn}} \big \},
\end{align}
where $\coneleq$ denotes elementwise inequality for vectors.
\item[(b)] Furthermore, for all iterations $\Time = 1, 2, \ldots$, we
  have
\begin{align}
\Expt \big[ \MASSERR{\ItErrNo{\Time}} \big] \; & \coneleq \; \Big( 6
\: \sum_{j=1}^{\dimn}\bound_j^2 \Big) \; \frac{(I - \log{\Time} \:
  \NilMat)^{-1}}{\Time} \: (\NilMat \: \onevec \: + \: 16)\, + \,
\NilMat (I-\NilMat)^{-1} \MASSERR{\GoodApp{\dimn}}.
\end{align}
\end{enumerate}
\end{theorem}

To clarify the statement in part (a), it guarantees that the
difference $\MASSERR{\ItErrNo{\Time}} - \Pi_\ball \big(
\MASSERR{\ItErrNo{\Time}} \big)$ between the error vector and its
projection onto the set $\ball$ converges almost surely to zero.  Part
(b) provides a quantitative guarantee on how quickly the expected
absolute value of this difference converges to zero.  In particular,
apart from logarithmic factors in $\Time$, the convergence rate
guarantees is of the order $\order(1/\Time)$.


\subsection{Bounds for general graphs}

Our next theorem addresses the case of general graphical models.  The
behavior of the ordinary BP algorithm to a graph with cycles---in
contrast to the tree-structured case---is more complicated.  On one
hand, for strictly positive potential functions (as considered in this
paper), a version of Brouwer's fixed point theorem can be used to
established existence of fixed points~\cite{WaiJorBook08}.  However,
in general, there may be multiple fixed points, and convergence is not
guaranteed.  Accordingly, various researchers have studied conditions
that are sufficient to guarantee uniqueness of fixed points and/or
convergence of the ordinary BP algorithm: one set of sufficient
conditions, for both uniqueness and convergence, involve assuming that
the BP update operator is a contraction in a suitable norm
(e.g.,~\cite{Tatikonda02,Ihler05,MooijKapp07,RoostaEtal08}).

In our analysis of the \ALG algorithm for a general graph, we impose
the following form of \emph{contractivity}: there exists a constant $0
< \contfact < 2$ such that
\begin{align}
\label{EqnContraction}
\norm{\upfun_{\LowerDirEdge{\fnode}{\snode}}(\mes) -
  \upfun_{\LowerDirEdge{\fnode}{\snode}}(\mes^{\prime})} \; & \leq \;
\big(1 - \frac{\contfact}{2} \big) \: \sqrt{ \frac{1}{|\Neig(\fnode)
    \backslash \{\snode\}|} \sum_{\CompNeig}
  \norm{\mes_{\LowerDirEdge{\tnode}{\fnode}} -
  \mes_{\LowerDirEdge{\tnode}{\fnode}}^{\prime}}^2},
\end{align}
for all directed edges $\DirEdge{\fnode}{\snode} \in \DirSet$, and feasible
messages $\mes$, and $\mes^{\prime}$.  We say that the ordinary BP
algorithm is $\contfact$-contractive when
condition~\eqref{EqnContraction} holds. \\

\begin{theorem}
\label{ThmGeneral}
Suppose that the ordinary BP algorithm is
$\contfact$-contractive~\eqref{EqnContraction}, and consider the
sequence of messages $\{\mes^{\Time}\}_{\Time = 0}^{\infty}$ generated
with step-size $\step = 1/(\contfact (\Time+1))$.  Then for all $\Time
= 1, 2, \ldots$, the error sequence $\{\ItErr{\Time}
\}_{\Time=0}^\infty$ is bounded in mean-square as
\begin{align}
\label{EqnGeneralBound}
\Expt \big[ \MASSERR{\ItErrNo{\Time}} \big] \; & \coneleq \; \biggr[
  \bigg(\frac{8 \sum_{j=1}^{\dimn} \bound_j^2}{\contfact^2} \bigg) \:
  \frac{\log \Time}{\Time} + \frac{1}{\contfact}
  \max_{\DirEdge{\fnode}{\snode} \in \DirSet}
  \|\GoodAppDown{\dimn}\|_{L^2}^2 \biggr] \, \onevec.
\end{align}
where $\GoodAppDown{\dimn} = \mes^*_{\LDE} - \Pir(\mes^*_{\LDE})$
is the approximation error on edge $\DirEdge{\fnode}{\snode}$.
\end{theorem}

Theorem~\ref{ThmGeneral} guarantees that under the contractivity
condition~\eqref{EqnContraction}, the \ALG iterates converge to a
neighborhood of the BP fixed point. The error offset depends on the
approximation error term that decays to zero as $\dimn$ is increased.
Moreover, disregarding the logarithmic factor, the convergence rate is
$\order(1/\Time)$, which is the best possible for a stochastic
approximation scheme of this type~\cite{NemYu83,AgaBarRavWai12}.


\subsection{Explicit rates for kernel classes}
\label{SecKernel}

Theorems~\ref{ThmTree} and~\ref{ThmGeneral} are generic results that
apply to any choices of the edge potential functions.  In this
section, we pursue a more refined analysis of the number of arithmetic
operations that are required to compute a \emph{$\delta$-uniformly
  accurate} approximation to the BP fixed point $\mes^*$, where
$\delta > 0$ is a user-specified tolerance parameter.  By a
$\delta$-uniformly accurate approximation, we mean a collection of
messages $\mes$ such that
\begin{align}
\label{EqnDeltaAcc}
\max_{\DirEdge{\fnode}{\snode} \in \DirSet} \Exs \big[ \|\mes_{\LDE} -
  \mes^*_{\LDE}\|_{L^2}^2 \big] & \leq \delta.
\end{align}
In order to obtain such an approximation, we need to specify both the
number of coefficients $\dimn$ to be retained, and the number of
iterations that we should perform.  Based on these quantities, our
goal is to the specify the \emph{minimal number of basic arithmetic
  operations} $T(\delta)$ that are sufficient to compute a
$\delta$-accurate message appproximation.

In order to obtain concrete answers, we study this issue in the
context of kernel-based potential functions.  In many applications,
the edge potentials $\epot: \Space \times \Space \rightarrow \real_+$
are symmetric and positive semidefinite (PSD) functions, frequently
referred to as kernel functions.\footnote{In detail, a PSD kernel
  function has the property that for all natural numbers $n$ and
  $\{x_1, \ldots, x_n \} \subset \Space$, the $n \times n$ kernel
  matrix with entries $\epot(x_i, x_j)$ is symmetric and positive
  semidefinite.} Commonly used examples include the Gaussian kernel
$\epot(x,y) = \exp(-\gamma\|x - y\|_2^2)$, the closely related
Laplacian kernel, and other types of kernels that enforce smoothness
priors.  Any kernel function defines a positive semidefinite integral
operator, namely via equation~\eqref{EqnIntegralUpdate}.  When
$\Space$ is compact and the kernel function is continuous, then
Mercer's theorem~\cite{RieNagBook90} guarantees that this integral
operator has a countable set of eigenfunctions
$\{\basis_j\}_{j=1}^\infty$ that form an orthonormal basis of
$L^2(\Space; \measure)$.  Moreover, the kernel function has the
expansion
\begin{align}
\label{EqnEigenDecomp}
\epot (x, y) \; = \; \sum_{j=1}^{\infty} \eigen_{j} \: \basis_j(x) \:
\basis_j(y),
\end{align}
where $\eigen_1 \geq \eigen_2 \geq \cdots \geq 0$ are the eigenvalues,
all guaranteed to be non-negative.  In general, the eigenvalues might
differ from edge to edge, but we suppress this dependence for
simplicity in exposition.  We study kernels that are trace class,
meaning that the eigenvalues are absolutely summable (i.e.,
$\sum_{j=1}^\infty \eigen_j < \infty$).

For a given eigenvalue sequence $\{\myeig_j\}_{j=1}^\infty$ and some
tolerance $\delta > 0$, we define the \emph{critical dimension}
$\dimcrit = \dimcrit(\delta; \{\myeig_j\})$ to be the smallest
positive integer $\dimn$ such that
\begin{align}
\label{EqnDefRstar}
\myeig_{\dimn} & \leq \delta.
\end{align}
Since $\myeig_j \rightarrow 0$, the existence of $\dimcrit < + \infty$
is guaranteed for any $\delta > 0$.

\begin{theorem}
\label{ThmKernel}
In addition to the conditions of Theorem~\ref{ThmGeneral}, suppose
that the compatibility functions are defined by a symmetric PSD
trace-class kernel with eigenvalues $\{\lambda_j\}$. If we run the
\ALG algorithm with $\dimcrit = \dimcrit(\delta; \{\myeig_j\})$ basis
coefficients, then it suffices to perform
\begin{align}
\label{EqnNumOpBound}
\NumOp(\delta; \{\myeig_j\}) \; & = \; \order \Big( \dimcrit \:
\big(\sum_{j=1}^{\dimcrit}\myeig_j^2\big) \; \big(1/\delta\big) \;
\log(1/\delta) \Big)
\end{align}
arithmetic operations per edge in order to obtain a $\delta$-accurate
message vector $\mes$.
\end{theorem}
\noindent The proof of Theorem~\ref{ThmKernel} is provided in
Section~\ref{SecProofThmKernel}.  It is based on showing that the
choice~\eqref{EqnDefRstar} suffices to reduce the approximation error
to $\order(\delta)$, and then bounding the total operation complexity
required to also reduce the estimation error.  \\

Theorem~\ref{ThmKernel} can be used to derive explicit estimates of
the complexity for various types of kernel classes.  We begin with the
case of kernels in which the eigenvalues decay at an inverse polyomial
rate: in particular, given some $\polydecay > 1$, we say that they
exhibit \emph{$\polydecay$-polynomial decay} if there is a universal
constant $\const$ such that 
\begin{align}
\label{EqnDefnPolyDecay}
\myeig_j \leq \const/j^{\polydecay} \quad \mbox{for all $j = 1, 2,
  \ldots$.}
\end{align}
Examples of kernels in this class include Sobolov spline
kernels~\cite{Gu02}, which are a widely used type of smoothness prior.
For example, the spline class associated with functions that are
$s$-times differentiable satisfies the decay
condition~\eqref{EqnDefnPolyDecay} with $\polydecay = 2 s$.

\begin{corollary}
\label{CorPolyDecay}
In addition to the conditions of Theorem~\ref{ThmGeneral}, suppose
that the compatibility functions are symmetric kernels with
$\polydecay$-polynomial decay~\eqref{EqnDefnPolyDecay}.  Then it
suffices to perform
\begin{align}
\NumOpPoly(\delta) & = \order \Big( \big(1/\delta \big)^{\frac{1 +
    \polydecay}{\polydecay}} \log(1/\delta) \Big)
\end{align}
operations per edge in order to obtain a $\delta$-accurate message
vector $\mes$.
\end{corollary}
\noindent The proof of this corollary is immediate from
Theorem~\ref{ThmKernel}: given the assumption $\lambda_j \leq
\const/j^{\polydecay}$, we see that $\dimcrit \leq
(\const/\delta)^{\frac{1}{\polydecay}}$ and
$\sum_{j=1}^{\dimcrit}\myeig_j^2 = \order(1)$.  Substituting into the
bound~\eqref{EqnNumOpBound} yields the claim.
Corollary~\ref{CorPolyDecay} confirms a natural intuition---namely,
that it should be easier to compute an approximate BP fixed point for
a graphical model with smooth potential functions.  Disregarding the
logarithmic factor (which is of lower-order), the operation complexity
$\NumOpPoly(\delta)$ ranges ranges from $\order \big( (1/\delta)^2
\big)$, obtained as $\alpha \rightarrow 1^+$ all the way down to
$\order \big(1/\delta \big)$, obtained as $\alpha \rightarrow
+\infty$. \\

Another class of widely used kernels are those with exponentially
decaying eigenvalues: in particular, for some $\polydecay > 0$, we say
that the kernel has \emph{$\polydecay$-exponential decay} if there are
universal constants $(C, c)$ such that 
\begin{align}
\label{EqnDefnExpDecay}
\myeig_j & \leq C \, \exp(-c j^{\polydecay}) \qquad \mbox{ for all $j
  = 1, 2, \ldots$.}
\end{align}
Examples of such kernels include the Gaussian kernel, which satisfies
the decay condition~\eqref{EqnDefnExpDecay} with $\polydecay = 2$
(e.g.,~\cite{SteChr08}).

\begin{corollary}
\label{CorExpDecay}
In addition to the conditions of Theorem~\ref{ThmGeneral}, suppose
that the compatibility functions are symmetric kernels with
$\polydecay$-exponential decay~\eqref{EqnDefnExpDecay}.  Then it
suffices to perform
\begin{align}
\label{EqnExpDecay}
\NumOp_{\operatorname{\tiny{exp}}}(\delta) & = \order \Big( (1/\delta)
\; \big(\log(1/\delta) \big)^{\frac{1 + \polydecay}{\polydecay}}
\Big).
\end{align}
operations per edge in order to obtain a uniformly $\delta$-accurate
message vector $\mes$.
\end{corollary}

As with our earlier corollary, the proof of this claim is a
straightforward consequence of Theorem~\ref{ThmKernel}.
Corollary~\ref{CorExpDecay} demonstrates that kernel classes with
exponentially decaying eigenvalues are not significantly different
from parametric function classes, for which a stochastic algorithm
would have operation complexity $\order(1/\delta)$.  Apart from the
lower order logarithmic terms, the complexity
bound~\eqref{EqnExpDecay} matches this parametric rate.


\section{Experimental Results}
\label{SecSimulations}

In this section, we describe some experimental results that help to
illustrate the theoretical predictions of the previous section.


\subsection{Synthetic Data}

\begin{figure}[t]
\begin{center}
\widgraph{0.5\textwidth}{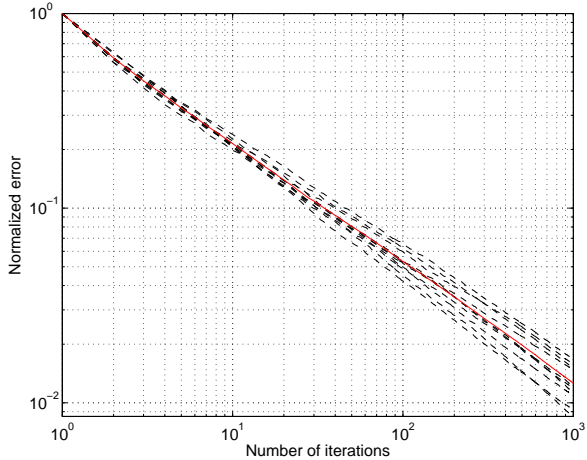}  
\caption{Plot of normalized error $\error^{\Time}/\error^{0}$ vs. the
  number of iterations $\Time$ for 10 different sample paths on a
  chain of size \mbox{$\numnode = 100$}. The dashed lines are sample
  paths whereas the solid line is the mean square error. In this
  experiment node and edge potentials are mixtures of three
  Gaussians~\eqref{EqnMix} and we implemented \ALG using the first
  $\dimn = 10$ Fourier coefficients with $\numsampY = 5$
  samples.} \label{FigSamplePath}
\end{center}
\end{figure}

We begin by running some experiments for a simple model, in which both
the node and edge potentials are mixtures of Gaussians.  More
specifically, we form a graphical model with potential functions of
the form
\begin{subequations}
\label{EqnMix}
\begin{align}
\npot(x_\snode) \; = & \sum_{i=1}^{3} \pi_{\snode;i} \exp
\big(-(x_\snode - \mu_{\snode;i})^2 / (2\sigma_{\snode;i}^2)\big),
\quad \mbox{for all $\snode \in \Vertex$, and} \\
\epot(x_\snode, x_\fnode) \; = & \; \sum_{i=1}^{3}
\pi_{\snode\fnode;i} \exp \big(-(x_\fnode -
x_\snode)^2/(2\sigma_{\snode\fnode;i}^2)\big) \quad \mbox{for all
  $(\snode, \fnode) \in \Edge$,}
\end{align}
\end{subequations}
where the non-negative mixture weights are normalized (i.e.,
$\sum_{i=1}^{3} \pi_{\snode\fnode;i} = \sum_{i=1}^{3} \pi_{\snode;i} =
1$). For each vertex and edge and for all $i = 1, 2, 3$, the mixture
parameters are chosen randomly from uniform distributions over the
range \mbox{$\sigma_{\snode;i}^2, \sigma_{\snode\fnode;i}^2 \in (0,
  0.5]$} and \mbox{$\mu_{\snode;i} \in [-3, 3]$}.

For a chain-structured graph with $\numnode = 100$ nodes, we first
compute the fixed point of standard BP, using direct numerical
integration to compute the integrals,\footnote{In particular, we
  approximate the integral update~\eqref{EqnBPUpdateEntry} with its
  Riemann sum over the range \mbox{$\Space = [-5,5]$} and with 100
  samples per unit time.}  so to compute (an extremely accurate
approximation of) the fixed point $\mes^{\ast}$.  We compare this
``exact'' answer to the approximation obtained by running the \ALG
algorithm using the first $\dimn = 10$ Fourier basis coefficients and
$\numsampY = 5$ samples.  Having run the \ALG algorithm, we compute
the average squared error
\begin{align}
\label{EqnDefnMSE}
\error^{\Time} & \defn \frac{1}{ \Dimn}
\sum_{\DirEdge{\fnode}{\snode}\in\DirSet} \sum_{j=1}^{\dimn}
(\mescoef_{\LDE; j}^{\Time} - \mescoef_{\LDE; j}^{\ast})^2
\end{align}
at each time $\Time = 1, 2, \ldots$.

\begin{figure}[h]
\begin{center}
\begin{tabular}{cc}
\widgraph{0.45\textwidth}{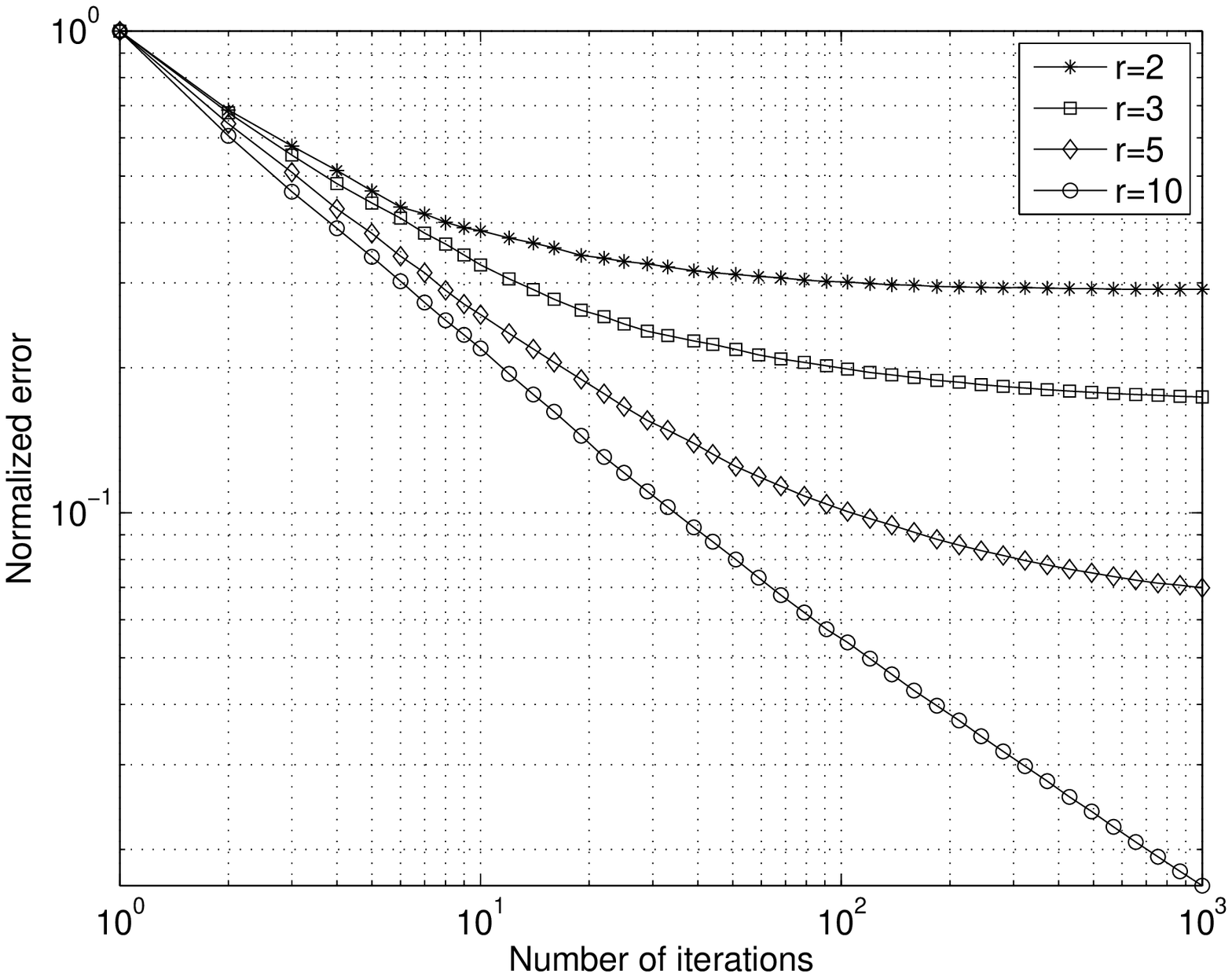} & 
\widgraph{0.45\textwidth}{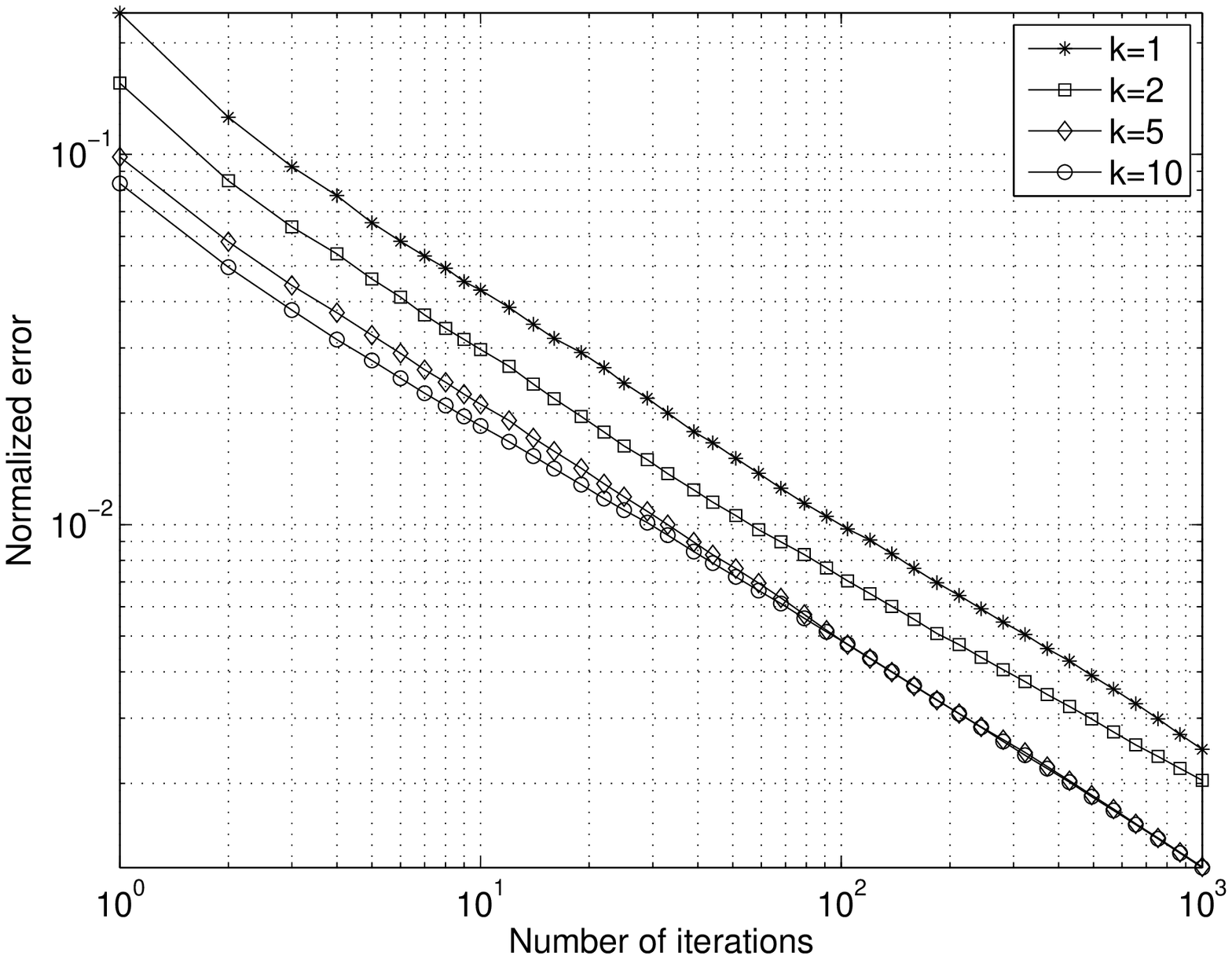} \\
(a) & (b)
\end{tabular}
\caption{ Normalized mean squared error
  $\Expt[\error^{\Time}/\error^{0}]$ verses the number of iterations
  for a Markov chain with $\numnode = 100$ nodes, using potential
  functions specified by the mixture of Gaussians
  model~\eqref{EqnMix}.  (a) Behavior as the number of expansion
  coefficients is varied over the range $\dimn \in \{2, 3, 5, 10\}$
  with $\numsampY = 5$ samples in all cases.  As predicted by the
  theory, the error drops monotonically with the number of iterations
  until it hits a floor. The error floor, which corresponds to the
  approximation error incurred by message expansion truncation,
  decreases as the number of coefficients $\dimn$ in increased. (b)
  Mean squared error $\Expt[\error^{\Time}]$ verses the number of
  iterations $\Time$ for different number of samples \mbox{$\numsampY
    \in \{1, 2, 5, 10\}$,} in all cases using $\dimn = 10$
  coefficients. Increasing the number of samples $\numsampY$ results
  in a downward shift in the error.  } \label{FigErrDiffExpCoef}
\end{center}
\end{figure}

Figure~\ref{FigSamplePath} provides plots of the
error~\eqref{EqnDefnMSE} versus the number of iterations for $10$
different trials of the \ALG algorithm.  (Since the algorithm is
randomized, each path is slightly different.)  The plots support our
claim of of almost sure convergence, and moreover, the straight lines
seen in the log-log plots confirm that convergence takes place at a
rate inverse polynomial in $\Time$.

In the next few simulations, we test the algorithm's behavior with
respect to the number of expansion coefficients $\dimn$, and number of
samples $\numsampY$. In particular, Figure~\ref{FigErrDiffExpCoef}(a)
illustrates the expected error, averaged over several sample paths,
vs. the number of iterations for different number of expansion
coefficients $\dimn \in\{ 2, 3, 5, 10\}$ when $\numsampY = 5$ fixed;
whereas Figure~\ref{FigErrDiffExpCoef}(b) depicts the expected error
vs. the number of iterations for different number of samples
$\numsampY \in \{1, 2, 5, 10\}$ when $\dimn = 10$ is fixed. As
expected, in Figure~\ref{FigErrDiffExpCoef}(a), the error decreases
monotonically, with the rate of $1/\Time$, till it hits a floor
corresponding the offset incurred by the approximation
error. Moreover, the error floor decreases with the number of
expansion coefficients. On the other hand, in
Figure~\ref{FigErrDiffExpCoef}(b), increasing the number of samples
causes a downward shift in the error. This behavior is also expected
since increasing the number of samples reduces the variance of the
empirical expectation in equation~\eqref{EqnUpdateInov}.


In our next set of experiments, still on a chain with $\numnode = 100$
vertices, we test the behavior of the \ALG algorithm on graphs with
edge potentials of varying degrees of smoothness.  In all cases, we
use node potentials from the Gaussian mixture ensemble~\eqref{EqnMix}
previously discussed, but form the edge potentials in terms of a
family of kernel functions.  More specifically, consider the basis
functions
\begin{align*}
\basis_{j}(x) & = \sin \big( (2j-1)\pi(x+5)/10 \big) \quad \mbox{for
  $j = 1, 2, \ldots$.}
\end{align*}
each defined on the interval $[-5, 5]$.  It is straightforward that
the family $\{\basis_j\}_{j=1}^\infty$ forms an orthonormal basis of
$L^2[-5, 5]$.  We use this basis to form the edge potential functions
\begin{align}
\label{EqnFinKer}
\epot (x, y) = \sum_{j=1}^{1000} (1/j)^{\smooth} \basis_{j}(x) \:
\basis_{j}(y),
\end{align}
where $\smooth > 0$ is a parameter to be specified.  By construction,
each edge potential is a positive semidefinite kernel function
satisfying the $\smooth$-polynomial decay
condition~\eqref{EqnDefnPolyDecay}.

Figure~\ref{FigSmoothKernel} illustrate the error curves for two
different choices of the smoothness parameter: panel (a) shows
$\smooth = 0.1$, whereas panel (b) shows $\smooth = 1$.  For the
larger value of $\smooth$ shown in panel (b), the messages in the BP
algorithm are smoother, so that the \ALG estimates are more accurate
with the same number of expansion coefficients. Moreover, similar to
what we have observed previously, the error decays with the rate of
$1/\Time$ till it hits the error floor.  Note that this error floor is
lower for the smoother kernel ($\smooth = 1$) compared to the rougher
case ($\smooth = 0.1$); note the difference in axis scaling between
panels (a) and (b).  Moreover, as predicted by our theory, the
approximation error decays faster for the smoother kernel, as shown by
the plots in Figure~\ref{FigErrorOffset}, in which we plot the final
error, due purely to approximation effects, versus the number of
expansion coefficients $\dimn$.  The semilog plot of
Figure~\ref{FigErrorOffset} shows that the resulting lines have
different slopes, as would be expected.

\begin{figure}
\begin{center}
\begin{tabular}{cc}

\widgraph{0.45\textwidth}{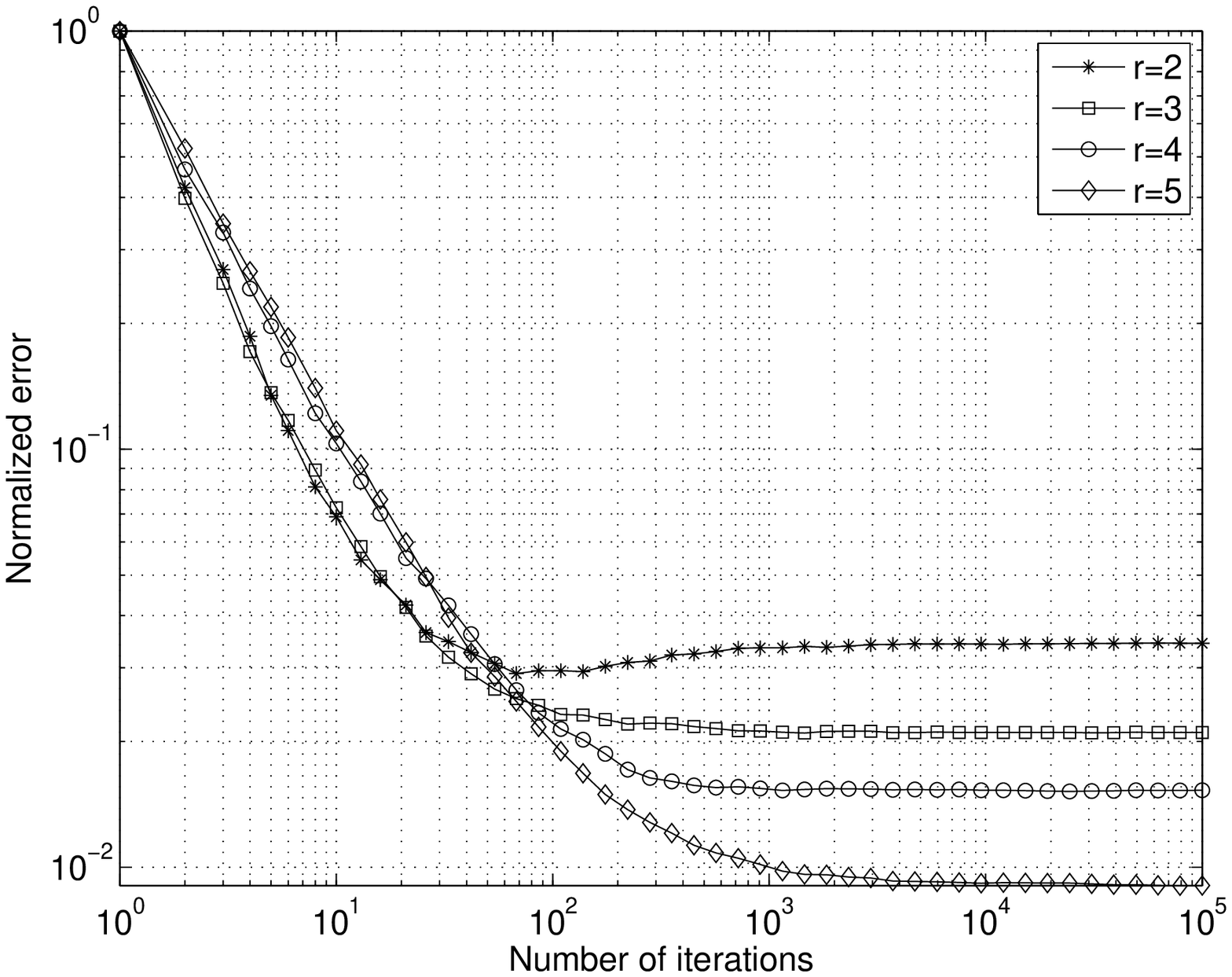} & 
\widgraph{0.45\textwidth}{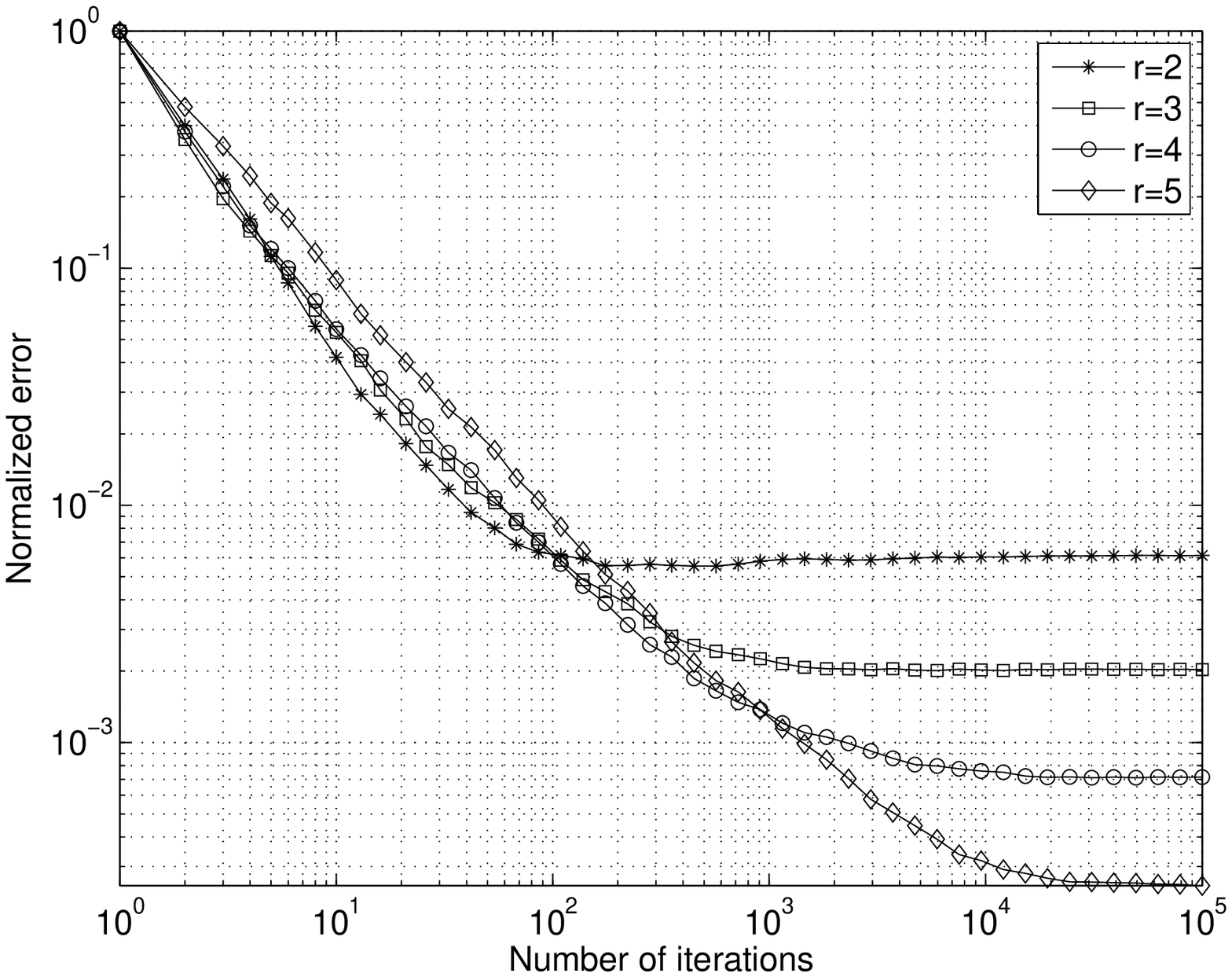} \\ 
(a) & (b)
\end{tabular}
\caption{Plot of the estimation error $\error^{\Time} / \error^{0}$
  verses the number of iterations $\Time$ for the cases of (a)
  $\smooth = 0.1$ and (b) $\smooth = 1$. The BP messages are smoother
  when $\smooth=1$, and accordingly the \ALG estimates are more
  accurate with the same number of expansion coefficients. Moreover,
  the error decays with the rate of $1/\Time$ till it hits a floor
  corresponding to the approximation error incurred by truncating the
  message expansion coefficients.} \label{FigSmoothKernel}
\end{center}
\vspace{0in}
\end{figure}

\begin{figure}
\begin{center}
\widgraph{0.6\textwidth}{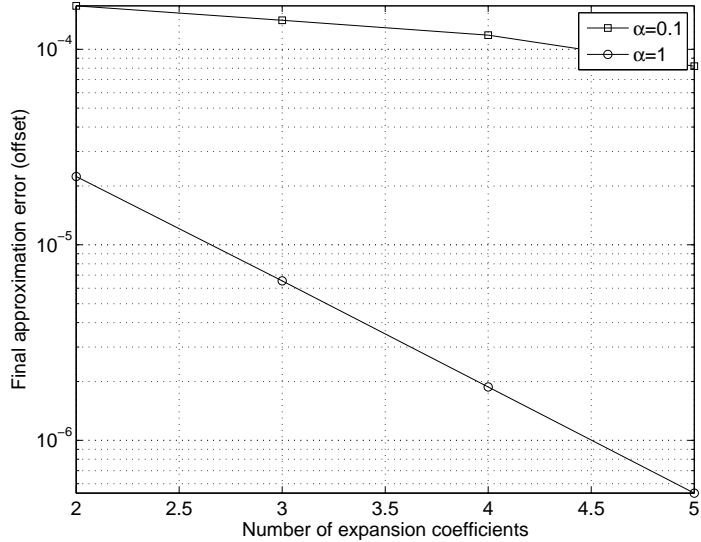}  
\caption{Final approximation error vs. the number of expansion
  coefficients for the cases of $\smooth = 0.1$ and $\smooth = 1$. As
  predicted by the theory, the error floor decays with a faster pace
  for the smoother edge potential.} \label{FigErrorOffset}
\end{center}
\vspace{0in}
\end{figure}


\subsection{Computer Vision Application}

\begin{figure}[h]
\begin{center}
\begin{tabular}{ccc}
\widgraph{0.45\textwidth}{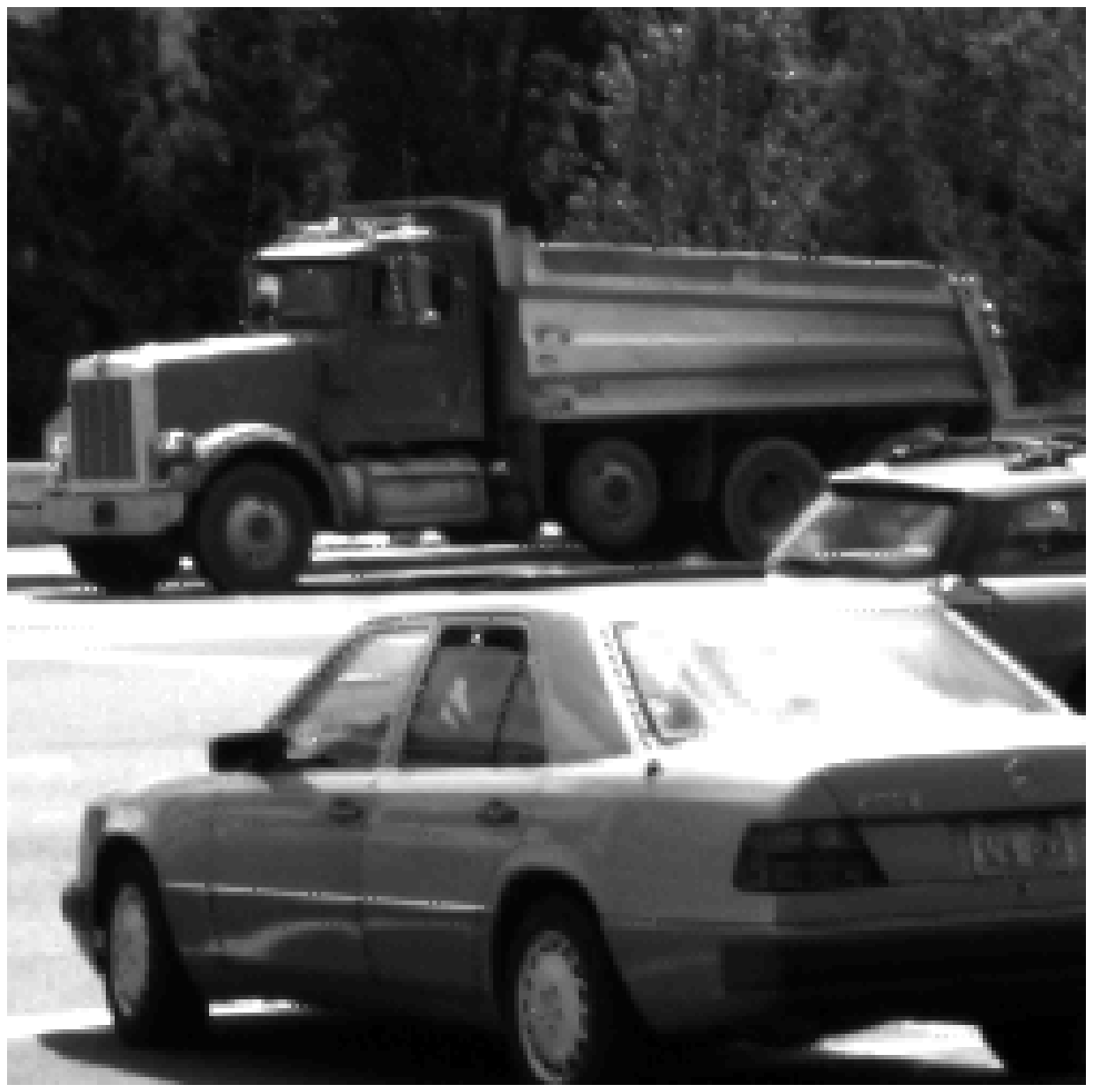}  & &
\widgraph{0.45\textwidth}{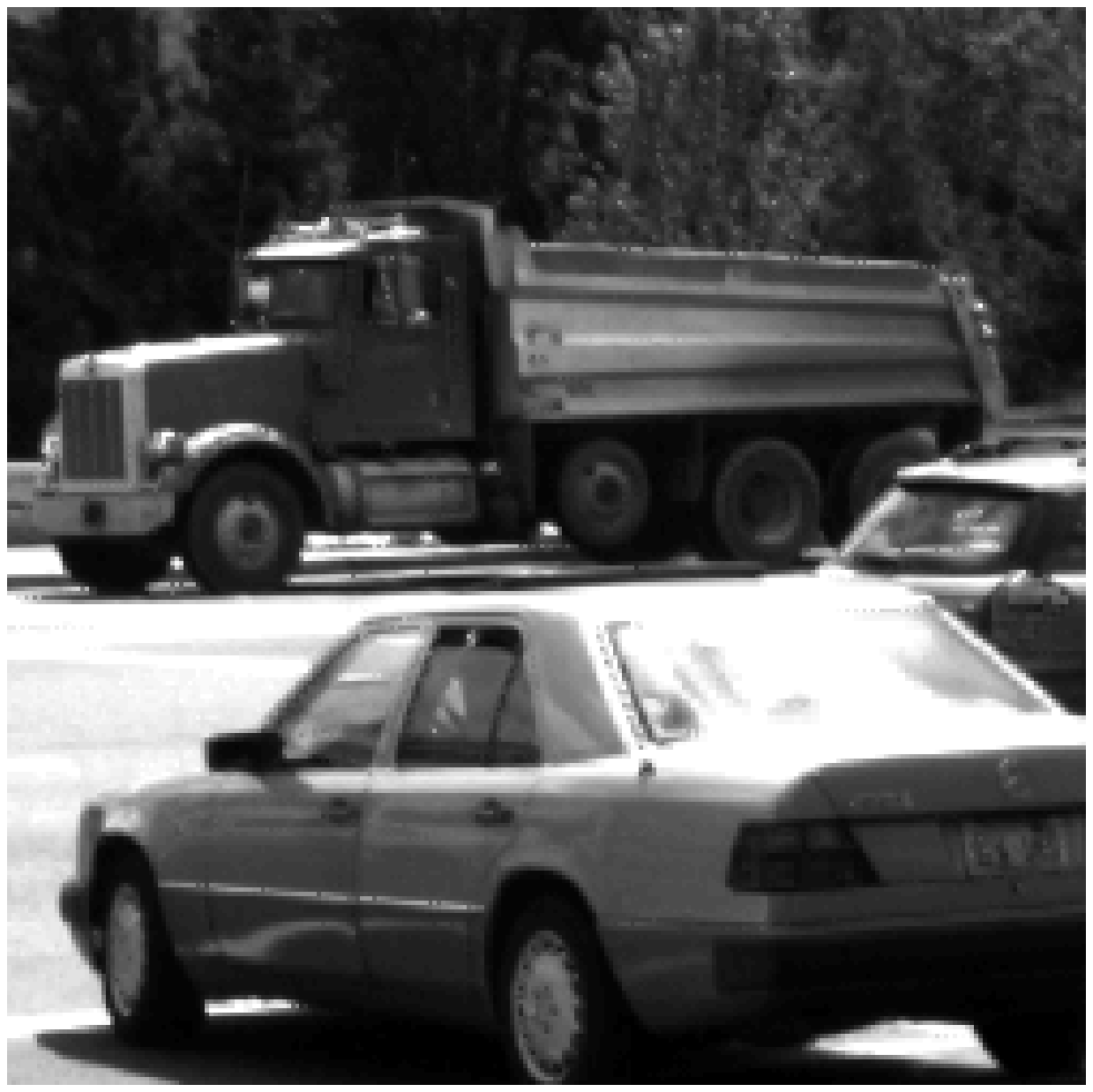}  \\
(a) & & (b)
\end{tabular}
\caption{Two frames, each of dimension $250 \times 250$ pixels, taken
  from a video sequence of moving cars.}
\label{FigOriginalFrame}
\end{center}
\vspace{0in}
\end{figure}

Moving beyond simulated problems, we conclude by showing the \ALG
algorithm in application to a larger scale problem that arises in
computer vision---namely, that of optical flow
estimation~\cite{Middlebury}.  In this problem, the input data are two
successive frames of a video sequence.  We model each frame as a
collection of pixels arrayed over a $\sqrt{n} \times \sqrt{n}$ grid,
and measured intensity values at each pixel location of the form
$\{\fframe(i, j), \sframe(i, j)\}_{i,j=1}^{\sqrt{n}}$.  Our goal is to
estimate a 2-dimensional motion vector $x_\snode = (x_{\snode;1},
x_{\snode;2})$ that captures the local motion at each pixel $\snode =
(i, j)$, $i, j = 1, 2, \ldots, \sqrt{\numnode}$ of the image sequence.

In order to cast this optical flow problem in terms of message-passing
on a graph, we adopt the model used by Boccignone et
al.~\cite{BocEtal07}.  We model the local motion $X_u$ as a
$2$-dimensional random vector taking values in the space $\Space =
[-d, d] \times [-d, d]$, and associate the random vector $X_\snode$
with vertex $\snode$, in a 2-dimensional grid (see
Figure~\ref{FigGraphicalModels}(a)).  At node $\snode = (i,j)$, we use
the change between the two image frames to specify the node potential
\begin{align*}
\compat_\snode(x_{\snode;1}, x_{\snode;2}) \; \propto \; \exp \bigg(
-\frac{(\fframe(i ,j) - \sframe(i + x_{\snode;1}, \; j + x_{\snode;2})
  )^2}{2\sigma_\snode^2} \bigg).
\end{align*}
On each edge $(\snode, \fnode)$, we introduce the potential function
\begin{align*}
\epot(x_\snode, x_\fnode) \; \propto \; \exp
\bigg(-\frac{\|x_\snode -
    x_\fnode\|^2}{2\sigma_{\snode\fnode}^2}\bigg),
\end{align*}
which enforces a type of \emph{smoothness prior} over the image.

To estimate the motion of a truck, we applied the \ALG algorithm using
the $2$-dimensional Fourier expansion as our orthonormal basis to two
$250 \times 250$ frames from a truck video sequence (see
Figure~\ref{FigOriginalFrame}).  We apply the \ALG algorithm using the
first $\dimn = 9$ coefficients and $\numsampY = 3$
samples. Figure~\ref{FigQSBPRecover} shows the HSV (hue, saturation,
value) codings of the estimated motions after $\Time = 1, 10, 40$
iterations, in panels (a), (b) and (c) respectively.  (Panel (d)
provides an illustration of the HSV encoding: hue is used to represent
in the angular direction of the motion whereas the speed (magnitude of
the motion) is encoded by the saturation (darker colors meaning higher
speeds).  The initial estimates of the motion vectors are noisy, but
it fairly rapidly converges to a reasonable optical flow field.  (To
be clear, the purpose of this experiment is not to show the
effectiveness of \ALG or BP as a particular method for optical flow,
but rather to demonstrate its correctness and feasibility of the \ALG
in an applied setting.)

\begin{figure}
\begin{center}
\begin{tabular}{cc}
  \widgraph{0.45\textwidth}{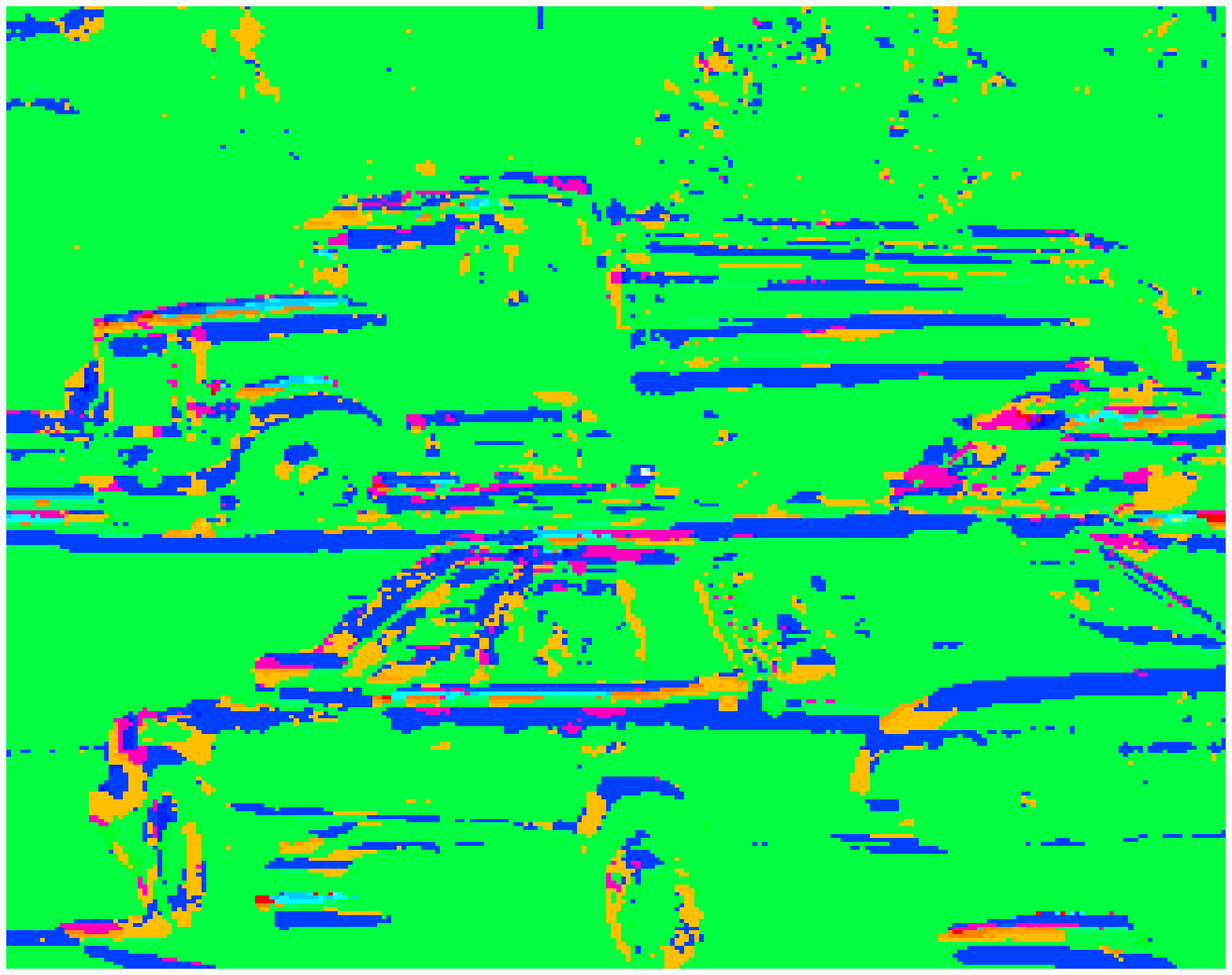} &
\widgraph{0.45\textwidth}{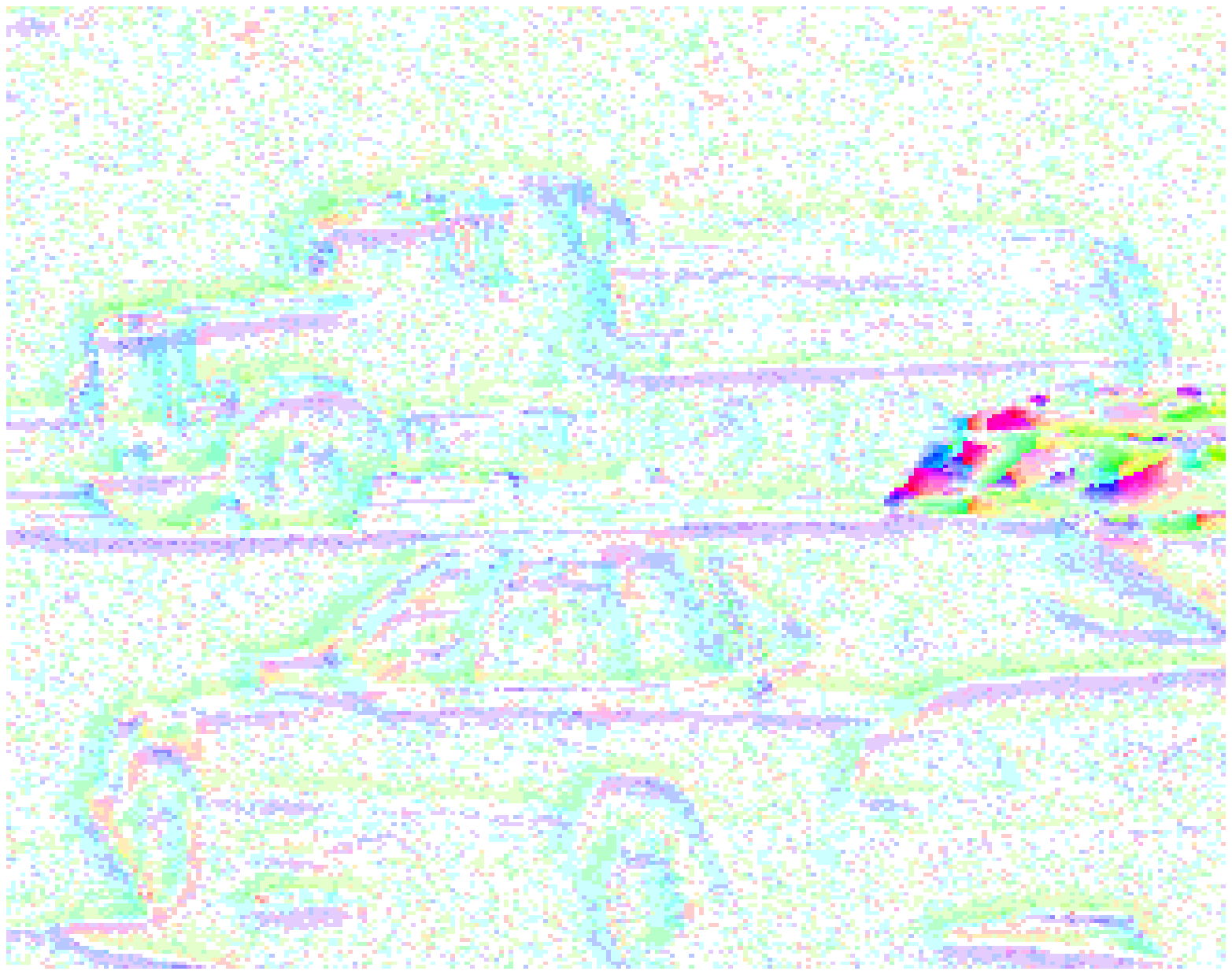} \\ (a) & (b)
\\ \widgraph{0.45\textwidth}{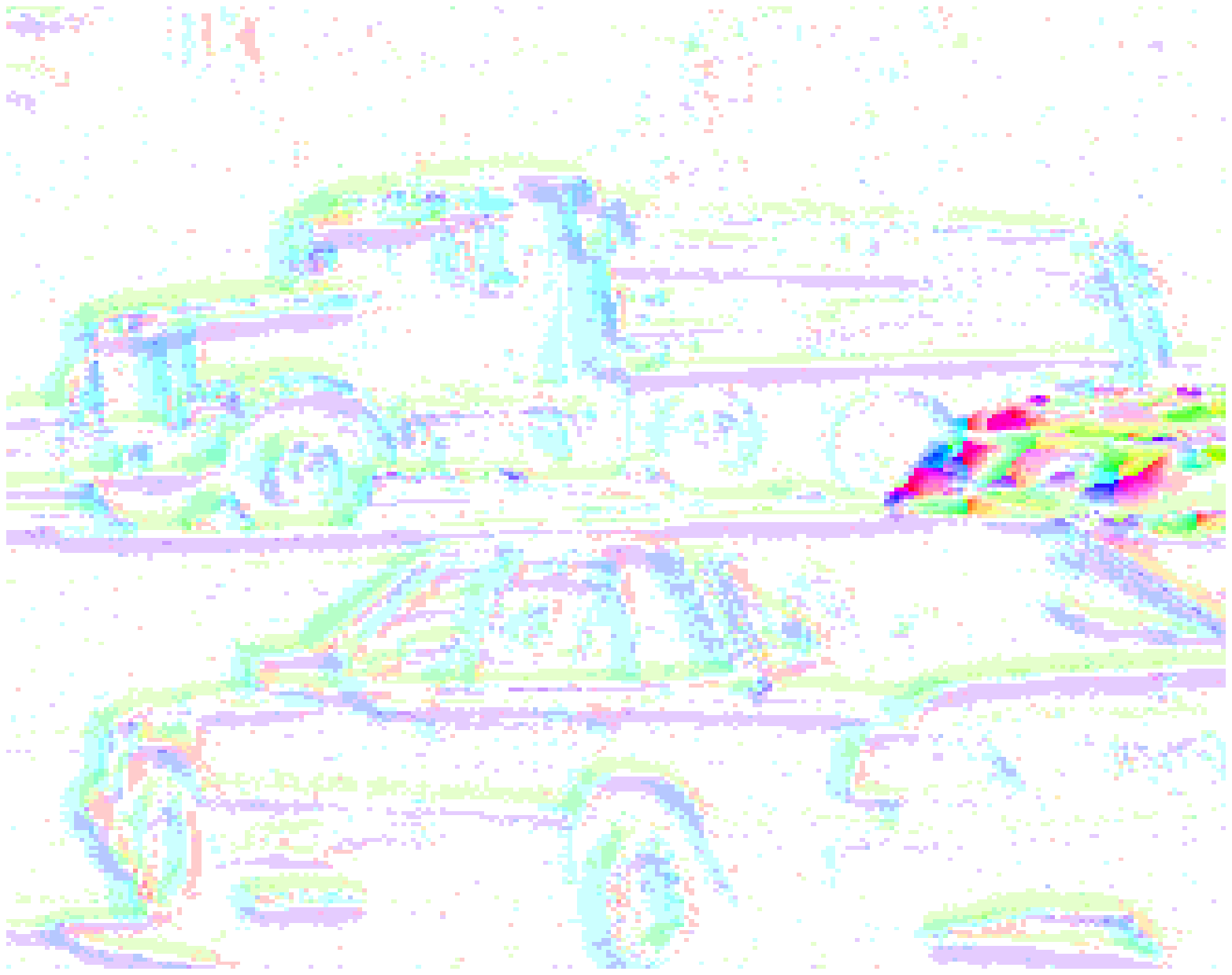} &
\raisebox{.4in}{\widgraph{0.25\textwidth}{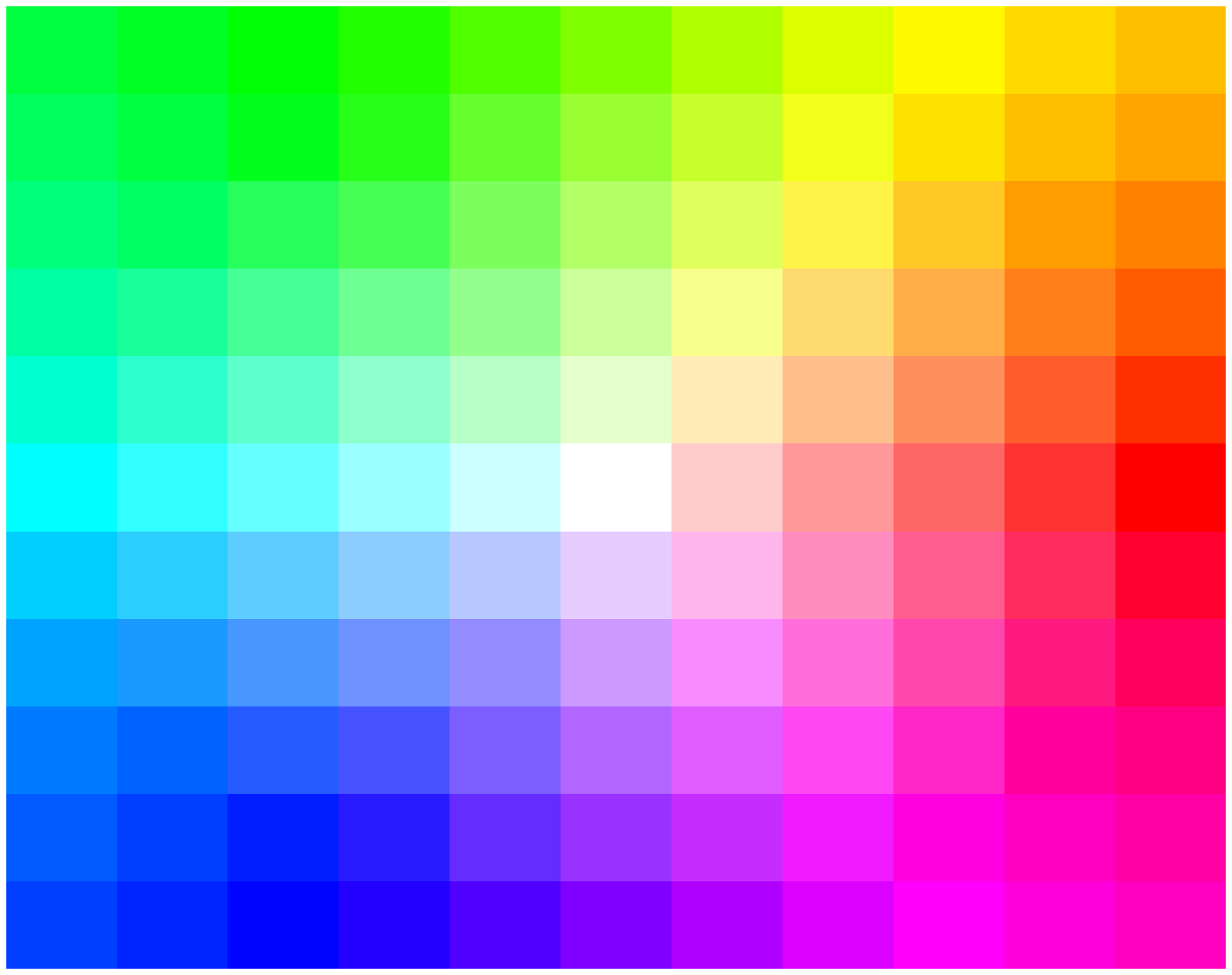}} \\ (c) & (d)
\end{tabular}
\caption{Color coded images of the estimated motion vectors after (a)
  $\Time = 1$, (b) $\Time = 10$, (c) $\Time = 40$ iterations. Panel
  (d) illustrates the hsv color coding of the flow. The color hue is
  used to encode the angular dimension of the motion, whereas the
  saturation level corresponds to the speed (length of motion
  vector). We implemented the \ALG algorithm by expanding in the
  two-dimensional Fourier basis, using $\dimn = 9$ coefficients and
  $\numsampY = 3$ samples. Although the initial estimates are noisy,
  it converges to a reasonable optical flow estimate after around $40$
  iterations.} \label{FigQSBPRecover}
\end{center}
\vspace{0in}
\end{figure}


\section{Proofs}
\label{SecProof}

We now turn to the proofs of our main results.  They involve a
collection of techniques from concentration of measure, stochastic
approximation, and functional analysis.


\subsection{Proof of Theorem~\ref{ThmTree}}
\label{SecThmTree}

Our goal is to bound the error
\begin{align}
\label{EqnPars}
\| \ItErr{\Time + 1}\|_{L^2}^2 \; & = \; \|\mesup{\Time+1} -
\Pir(\mesup{*})\|_{L^2}^2 \; = \; \sum_{j=1}^\dimn \big(
\coeffdirthree{\Time+1}{j} - \coeffbptwo{j} \big)^2,
\end{align}
where the final equality follows by Parseval's theorem.  Here
$\{\coeffbptwo{j} \}_{j=1}^\dimn$ are the basis expansion coefficients
that define the best $\dimn$ approximation to the BP fixed point
$\mes^*$.  The following lemma provides an upper bound on this error
in terms of two related quantities.  First, we let
$\{\compupfuncoef\}_{j=1}^\infty$ denote the basis function expansion
coefficients of the
$\upfun_{\LDE}(\meshat^t_{\LowerDirEdge{\fnode}{\snode}})$---that is,
$[\upfun_{\LDE}(\meshat^t_{\LowerDirEdge{\fnode}{\snode}})](\cdot) =
\sum_{j=1}^\infty \compupfuncoef \basis_j(\cdot)$.  Second, for each
$j = 1, 2, \ldots, \dimn$, define the deviation
$\deviatetwo{\Time+1}{j} \defn \coefftilthree{\Time+1}{j} -
\upfuncoeftwo{\Time}{j}$, where the coefficients
$\coefftilthree{\Time+1}{j}$ are updated in Step 2(c) Figure~\ref{FigAlg}.
\begin{lemma}
\label{LemTwoTerm}
For each iteration $\Time = 0, 1, 2, \ldots$, we have
\begin{align}
\label{EqnTwoTerm}
\| \ItErr{\Time+1}\|_{L^2}^2 \; & \leq \;
\underbrace{\frac{2}{\Time+1} \: \sum_{j=1}^\dimn
  \sum_{\newt=0}^{\Time} \: \big[ \upfuncoeftwo{\newt}{j} -
    \coeffbptwo{j} \big]^2}_{\mbox{Deterministic term $\detterm$}} \,
+ \, \underbrace{\frac{2}{(\Time+1)^2} \: \sum_{j=1}^\dimn \Big \{
  \sum_{\newt=0}^{\Time} \deviatetwo{\newt+1}{j} \Big
  \}^2}_{\mbox{Stochastic term $\stocterm$}}
\end{align}
\end{lemma}
\noindent The proof of this lemma is relatively straightforward; see
Appendix~\ref{AppLemTwoTerm} for the details.  Note that
inequality~\eqref{EqnTwoTerm} provides an upper bound on the error
that involves two terms: the first term $\detterm$ depends only on the
expansion coefficients $\{ \upfuncoeftwo{\newt}{j}, \newt = 0, \ldots,
\Time \}$ and the BP fixed point, and therefore is a deterministic
quantity when we condition on all randomness in stages up to step
$\Time$.  The second term $\stocterm$, even when conditioned on
randomness through step $\Time$, remains stochastic, since the
coefficients $\coefftil{\Time+1}$ (involved in the error term
$\deviate{\Time+1}$) are updated stochastically in moving from
iteration $\Time$ to $\Time+1$.

We split the remainder of our analysis into three parts: (a) control
of the deterministic component; (b) control of the stochastic term;
and (c) combining the pieces to provide a convergence bound.

\subsubsection{Upper-bounding the deterministic term}

By the Pythagorean theorem, we have
\begin{align}
\label{EqnBeforeLips} 
\sum_{\newt=0}^{\Time} \sum_{j=1}^\dimn \big[ \upfuncoeftwo{\newt}{j}
  - \coeffbptwo{j} \big]^2 \; & \leq \; \sum_{\newt = 0}^\Time \|
\upfun_{\LowerDirEdge{\fnode}{\snode}}(\meshat^\Time) -
\upfun_{\LowerDirEdge{\fnode}{\snode}}(\mes^*) \|^2_{L^2}
\end{align}
In order to control this term, we make use of the following lemma,
proved in Appendix~\ref{AppLemLipschitz}:
\begin{lemma}
\label{LemLipschitz}
For all directed edges $\DirEdge{\fnode}{\snode} \in \DirSet$, there
exist constants $\{ \Lips, \mbox{\CompNeig} \}$ such that
\begin{align*}
\|\upfun_{\LowerDirEdge{\fnode}{\snode}}(\meshat^\Time) \, - \,
\upfun_{\LowerDirEdge{\fnode}{\snode}}(\mes^*) \|_{L^2} \; \;& \leq
\sum_{\CompNeig} \Lips \: \|
\meshat^{\Time}_{\LowerDirEdge{\tnode}{\fnode}} \, - \,
\mes_{\LowerDirEdge{\tnode}{\fnode}}^{\ast} \|_{L^2}
\end{align*}
for all $\Time = 1, 2, \ldots$.
\end{lemma}

Substituting the result of Lemma~\ref{LemLipschitz} in
equation~\eqref{EqnBeforeLips} and performing some algebra, we find
that
\begin{align}
\sum_{\newt=0}^{\Time} \sum_{j=1}^\dimn \big[ \upfuncoeftwo{\newt}{j}
  - \coeffbptwo{j} \big]^2 \; & \leq \; \sum_{\newt = 0}^\Time \Big(
\sum_{\CompNeig} \Lips \: \|
\meshat^{\newt}_{\LowerDirEdge{\tnode}{\fnode}} \, - \,
\mes_{\LowerDirEdge{\tnode}{\fnode}}^{\ast} \|_{L^2} \Big)^2 \nonumber
\\
\label{EqnIntermediate}
\; & \leq \; (\degr_\fnode -1 ) \; \sum_{\newt = 0}^\Time \sum_{\CompNeig}
\Lips^2 \: \| \meshat^{\newt}_{\LowerDirEdge{\tnode}{\fnode}} \, - \,
\mes_{\LowerDirEdge{\tnode}{\fnode}}^{\ast} \|^2_{L^2},
\end{align}
where $\degr_{\fnode}$ is the degree of node $\fnode\in\Vertex$.  By
definition, the message
$\meshat^{\newt}_{\LowerDirEdge{\tnode}{\fnode}}$ is the
$L^2$-projection of $\mes^\newt_{\LowerDirEdge{\tnode}{\fnode}}$ onto
$\FunSpace$.  Since $\mes^{\ast}_{\LowerDirEdge{\tnode}{\fnode}}
\in \FunSpace$ and projection is non-expansive, we have
\begin{align}
\norm{\meshat^{\newt}_{\LowerDirEdge{\tnode}{\fnode}} \, - \,
\mes_{\LowerDirEdge{\tnode}{\fnode}}^{\ast}}^2 \; & \leq \; \norm{
\mes^{\newt}_{\LowerDirEdge{\tnode}{\fnode}} \, - \,
\mes_{\LowerDirEdge{\tnode}{\fnode}}^{\ast} }^2 \nonumber \\
\label{EqnIntermediateTwo}
\; & = \; \norm{\ItErrtwo{\newt}{\tnode\rightarrow\fnode}}^2 \, + \,
\norm{\GoodApptwo{\dimn}{\tnode\rightarrow\fnode}}^2
\end{align}
where in the second step we have used the Pythagorean identity and
recalled the definitions of estimation error as well as approximation error
from~\eqref{EqnDefEstErr} and~\eqref{EqnDefApproxErr}.  
%
%
%
%
Substituting the inequality~\eqref{EqnIntermediateTwo} into the
bound~\eqref{EqnIntermediate} yields
\begin{align*}
 \sum_{\newt=0}^{\Time}\sum_{j=1}^\dimn \big[ \upfuncoeftwo{\newt}{j}
   - \coeffbptwo{j} \big]^2 \; & \leq \; (\degr_\fnode -1 ) \,
 \sum_{\newt=0}^{\Time} \sum_{\CompNeig} \Lips^2 \, \big(
 \norm{\ItErrtwo{\newt}{\tnode\rightarrow\fnode}}^2 +
 \norm{\GoodApptwo{\dimn}{\tnode\rightarrow\fnode}}^2 \big).
\end{align*}
Therefore, introducing the convenient shorthand $\matentry \defn 2\:
(\degr_{\fnode} - 1) \: \Lips^2$, we have shown that
\begin{align}
\label{EqnDetermPart}
\detterm \; & \leq \; \frac{1}{\Time+1} \: \sum_{\newt=0}^{\Time} \;
\sum_{\CompNeig} \matentry \:
\big(\norm{\ItErrtwo{\Time}{\LowerDirEdge{\tnode}{\fnode}}}^2 \, + \,
\norm{\GoodApptwo{\dimn}{\LowerDirEdge{\tnode}{\fnode}}}^2 \big).
\end{align}
We make further use of this inequality shortly.


\subsubsection{Controlling the stochastic term}

We now turn to the stochastic part of the
inequality~\eqref{EqnTwoTerm}.  Our analysis is based on the following
fact, proved in Appendix~\ref{AppLemMartingale}:
\begin{lemma}
\label{LemMartingale}
For each $\Time \geq 0$, let $\Field^\Time \defn \sigma(\mes^0,
\ldots, \mes^\Time)$ be the $\sigma$-field generated by all messages
through time $\Time$. Then for every fixed $j=1,2, \ldots, \dimn$, the
sequence $\deviatetwo{\Time+1}{j} = \coefftilthree{\Time+1}{j} -
\upfuncoeftwo{\Time}{j}$ is a bounded martingale difference with
respect to $\{\Field^\Time\}_{\Time=0}^\infty$. In particular, we have
$|\deviatetwo{\Time+1}{j}| \leq 2 \bound_j$, where $\bound_j$ was
previously defined~\eqref{EqnDefnBound}.
\end{lemma}
\noindent 
Based on Lemma~\ref{LemMartingale}, standard martingale convergence
results~\cite{Durrett95} guarantee that for each $j = 1, 2, \ldots,
\dimn$, we have $ \sum_{\newt=0}^{\Time} \deviatetwo{\newt+1}{j} /
(\Time+1)$ converges to $0$ almost surely (a.s.) as $\Time \rightarrow
\infty$, and hence
\begin{align}
\label{EqnasConv}
\stocterm \; & = \; \frac{2}{(\Time+1)^2} \: \sum_{j = 1}^{\dimn} \: \biggr\{
\sum_{\newt=0}^{\Time} \deviatetwo{\newt+1}{j} \; \biggr\}^2 \; = \; 2 \:
\sum_{j = 1}^{\dimn} \: \biggr \{ \frac{1}{\Time+1}
\sum_{\newt=0}^{\Time} \deviatetwo{\newt+1}{j} \biggr \}^2 \;
\as 0.
\end{align}
Furthermore, we can apply the Azuma-Hoeffding
inequality~\cite{ChuLu06} in order to characterize the rate of
convergence.  For each $j = 1, 2, \ldots, \dimn$, define the
non-negative random variable \mbox{$Z_j \defn \big \{
  \sum_{\newt=0}^{\Time} \deviatetwo{\newt+1}{j} \big\}^2 /
  (\Time+1)^2$.}  Since $|\deviatetwo{\newt+1}{j}| \leq 2 \bound_j$, for
any $\delta \geq 0$, we have
\begin{align*}
\Prob \big[ Z_j \geq \delta \big] \; & = \; \Prob \big[ \sqrt{Z_j} \;
  \geq \; \sqrt{\delta} \big] \; \leq \; 2 \: \exp\bigg(-
\frac{(\Time+1) \; \delta }{8 \: \bound_j^2}\bigg),
\end{align*}
for all $\delta > 0$.  Moreover, $Z_j$ is non-negative; therefore,
integrating its tail bound we can compute the expectation
\begin{align*}
\Exs[Z_j] \; & = \; \int_0^\infty \mprob[ Z_j \geq \delta ] \: d \delta \; \leq \;
2 \int_0^\infty \exp\bigg(- \frac{(\Time+1) \; \delta }{8 \:
  \bound_j^2}\bigg) \: d \delta \; = \; \frac{16 \bound_j^2}{\Time+1},
\end{align*}
and consequently
\begin{align}
\label{EqnConRate}
\Expt [|\stocterm|] \; \le \; \frac{32 \:
  \sum_{j=1}^{\dimn}\bound_j^2}{\Time+1}.
\end{align}


\subsubsection{Establishing convergence}

We now make use of the results established so far to prove the claims.
Substituting the upper bound~\eqref{EqnDetermPart} on $\detterm$ into
the decomposition~\eqref{EqnTwoTerm} from Lemma~\ref{LemTwoTerm}, we
find that
\begin{align}
\label{EqnScalarInq}
\norm{\ItErr{\Time+1}}^2 \; & \leq \; \frac{1}{\Time+1} \:
\sum_{\newt=0}^{\Time} \sum_{\CompNeig} \matentry \: \big \{ \|
\ItErrNo{\newt}_{\LowerDirEdge{w}{\fnode}}\|_{L^2}^2 \, +\,
\norm{\GoodApptwo{\dimn}{\LowerDirEdge{\tnode}{\fnode}}}^2 \big \}
+ \stocterm.
\end{align}
For convenience, let us introduce the vector $\term^{\Time+1} = \{
\term^{\Time+1}_{\LowerDirEdge{\fnode}{\snode}}, \DirEdge{\fnode}{\snode}
\in \DirSet \} \in \real^{\Dimn}$ with entries
\begin{align}
\label{EqnSmallTerm}
\term^{\Time+1}_{\LowerDirEdge{\fnode}{\snode}} \; \defn & \;
\frac{1}{\Time+1} \: \Big \{ \sum_{\CompNeig}\!\!  \matentry \:
\norm{\ItErrNo{0}_{\LowerDirEdge{w}{\fnode}}}^2 \Big \} \, + \,
\stocterm.
\end{align}
Now define a matix $\NilMat \in \real^{\Dimn \times \Dimn}$ with
entries indexed by the directed edges and set to
\begin{align}
\label{EqnDefnNilmat}
\NilMat_{\LowerDirEdge{\fnode}{\snode}, \; \LowerDirEdge{w}{s}} &
\defn \; \begin{cases} \matentry & \mbox{if $\ftnode = \fnode$ and
    $\CompNeig$} \\ 0 & \mbox{otherwise.}
\end{cases}
\end{align}
In terms of this matrix and the error terms $\MASSERR{\cdot}$
previously defined in equations~\eqref{EqnDefnMassErr}
and~\eqref{EqnDefnMassErrApp}, the scalar
inequalities~\eqref{EqnScalarInq} can be written in the matrix form
\begin{align}
\label{EqnMatrixIneq}
\MASSERR{\ItErrNo{\Time+1}} \; & \coneleq \; \NilMat \,
\Big[\frac{1}{\Time+1} \: \sum_{\newt=1}^{\Time}
  \MASSERR{\ItErrNo{\newt}} \Big] + \, \NilMat \,
  \MASSERR{\GoodApp{\dimn}} \, + \, \term^{\Time+1},
\end{align}
where $\coneleq$ denotes the element-wise inequality based on the
orthant cone.

From Lemma 1 in the paper~\cite{NooWai12a}, the matrix $\NilMat$ is
guaranteed to be nilpotent with degree $\diam$ equal to the graph
diameter.  Consequently, unwrapping the
recursion~\eqref{EqnMatrixIneq} for a total of $\diam=
\diameter(\Graph)$ times yields
\begin{align*}
\MASSERR{\ItErrNo{\Time+1}} \; & \coneleq \; \term_{0}^{\Time+1} \, + \,
\NilMat \: \term_{1}^{\Time+1} \, + \, \ldots \, + \,
\NilMat^{\diam-1} \: \term_{\diam-1}^{\Time+1} \,+\, (\NilMat \, + \,
\NilMat^2 \, + \, \ldots \,+ \,\NilMat^{\diam}) \:
\MASSERR{\GoodApp{\dimn}},
\end{align*}
where we define $\term_{0}^{\Time+1} \equiv \term^{\Time+1}$, and then
recursively $\term_{s}^{\Time+1} \defn (\sum_{\newt=1}^{\Time}
\term_{s-1}^{\newt}) / (\Time+1)$ for \mbox{$s = 1, 2, \ldots,
  \diam-1$}. By the nilpotency of $\NilMat$, we have the identity
\mbox{$I + \NilMat + \ldots + \NilMat^{\diam-1} = (I-\NilMat)^{-1}$};
so we can further simplify the last inequality
\begin{align}
\label{EqnUnwrapMatIneq}
\MASSERR{\ItErrNo{\Time+1}} \; & \coneleq \; \sum_{s=0}^{\diam-1}
\NilMat^{s} \: \term_{s}^{\Time+1} \, + \, \NilMat\:(I-\NilMat)^{-1} \:
\MASSERR{\GoodApp{\dimn}}.
\end{align}
Recalling the definition $\ball \defn \big\{ e \in \real^{\Dimn} \,
\mid \, |e| \coneleq \NilMat(I-\NilMat)^{-1} \MASSERR{\GoodApp{\dimn}}
\big \}$, inequality~\eqref{EqnUnwrapMatIneq} implies that
\begin{align}
\label{EqnNearlyThere}
\big| \MASSERR{\ItErrNo{\Time+1}}\, - \, \project{
  \MASSERR{\ItErrNo{\Time+1}}} \big| & \coneleq \;
\sum_{s=0}^{\diam-1} \: \NilMat^{s} \: \term_{s}^{\Time+1}.
\end{align}
We now use the bound~\eqref{EqnNearlyThere} to prove both parts of
Theorem~\ref{ThmTree}.

\paragraph{Proof of Theorem~\ref{ThmTree}(a):} To 
prove the almost sure convergence claim in part (a), it suffices to
show that for each $s = 0, 1, \ldots, \diam-1$, we have
$\term_{s}^{\Time} \as 0$ as $\Time \rightarrow +\infty$.  From
equation~\eqref{EqnasConv} we know $\stocterm\to0$ almost surely as
$\Time\to\infty$. In addition, the first term in~\eqref{EqnSmallTerm}
is at most $\order(1/\Time)$, so that also converges to zero as
$\Time\to \infty$. Therefore, we conclude that $\term_0^{\Time} \as 0$
as $\Time \to \infty$.

In order to extend this argument to higher-order terms, let us recall
the following elementary fact from real analysis~\cite{Royden}: for
any sequence of real numbers $\{x^\Time\}_{\Time=0}^{\infty}$ such
that $x^{\Time}\to0$, then we also have $(\sum_{\newt = 0}^{\Time}
x^{\newt}) / \Time \to 0$.  In order to apply this fact, we observe
that $\term_0^\Time \as 0$ means that for almost every sample point
$\omega$ the deterministic sequence
$\{\term_0^{\Time+1}(\omega)\}_{\Time=0}^{\infty}$ converges to zero.
Consequently, the above fact implies that $\term_{1}^{\Time+1}
(\omega) = (\sum_{\newt=1}^{\Time} \term_{0}^{\newt} (\omega)) /
(\Time+1) \to 0$ as $\Time \to \infty$ for almost all sample points
$\omega$, which is equivalent to asserting that $\term_{1}^{\Time} \as
\vec{0}$. Iterating the same argument, we establish
$\term_{s}^{\Time+1} \as \vec{0}$ for all $s = 0, 1, \ldots, \diam-1$,
thereby concluding the proof of Theorem~\ref{ThmTree}(a).

\paragraph{Proof of Theorem~\ref{ThmTree}(b):}
Taking the expectation on both sides of the
inequality~\eqref{EqnUnwrapMatIneq} yields
\begin{align}
\label{EqnBoundConvRate}
\Expt \big[ |\MASSERR{\ItErrNo{\Time+1}} -
    \project{\MASSERR{\ItErrNo{\Time+1}}} | \big] \; & \coneleq \;
  \sum_{s=0}^{\diam-1} \NilMat^{s} \: \Expt [\term_{s}^{\Time+1}].
\end{align}
so that it suffices to upper bound the expectations $\Expt
[\term_{s}^{\Time+1}]$ for $s = 0, 1, \ldots, \diam-1$.  In
Appendix~\ref{AppLemUpperExp}, we prove the following result:
\begin{lemma}
\label{LemUpperExp}
 Define the $\Dimn$-vector $\vvec \defn \big \{ \sum_{j=1}^{\dimn}
 \bound_j^2 \big \} (2 \NilMat \onevec + 32 )$.  Then for all
 $s = 0, 1, \ldots, \diam-1$ and $\Time = 0, 1, 2, \ldots$,
\begin{align}
\label{EqnUpperExp}
\Expt [\term_{s}^{\Time+1}] \; \coneleq & \; \frac{\vvec}{\Time+1}
\: \bigg(\sum_{u=0}^{s} \frac{(\log(\Time+1))^{u}}{u!}\bigg),
\end{align}
\end{lemma}

Using this lemma, the proof of part (b) follows easily.  In
particular, substituting the bounds~\eqref{EqnUpperExp} into
equation~\eqref{EqnBoundConvRate} and doing some algebra yields
\begin{align*}
\Expt \big[ |\MASSERR{\ItErrNo{\Time+1}}\, - \,
  \project{\MASSERR{\ItErrNo{\Time+1}}} | \big] \; \coneleq &\;
\sum_{s=0}^{\diam-1} \NilMat^{s}
\sum_{u=0}^{s} \frac{(\log(\Time+1))^{u}}{u!} \, \Big(\frac{\vvec}{\Time+1}\Big)\\
\; \coneleq & \; 3 \: \sum_{s = 0}^{\diam-1}
(\log{(\Time+1)})^{s}\:\NilMat^{s } \, \Big(\frac{\vvec}{\Time+1} \Big)\\
\; \coneleq & \; 3 \: (I \, - \,
\log{(\Time+1)}\:\NilMat)^{-1} \, \Big(\frac{\vvec}{\Time+1}\Big),
\end{align*}
where again we used the fact that $\NilMat^{\diam} = 0$.


\subsection{Proof of Theorem~\ref{ThmGeneral}}
\label{SecThmGeneral}

Recall the definition of the estimation error $\ItErr{\Time}$
from~\eqref{EqnDefEstErr}. By Parseval's identity we know that
$\norm{\ItErr{\Time}}^2 = \sum_{j=1}^{\dimn} (\coeffupdown{\Time}{j} -
\coeffupdown{*}{j})^2$. For convenience, we introduce the following
shorthands for mean squared error on the directed edge
$\DirEdge{\fnode}{\snode}$
\begin{align*} 
\AVEERR{\ItErr{\Time}} \; \defn \; \Expt [\norm{\ItErr{\Time}}^2] \; =
\;\Expt \big[\sum_{j=1}^{\dimn} (\coeffupdown{\Time}{j} -
  \coeffupdown{*}{j})^2\big],
\end{align*}
as well as the $\ell_{\infty}$ error
\begin{align*} 
\RHOBARSQMAX{\ItErrNo{\Time}} \; \defn \;
\max_{\DirEdge{\fnode}{\snode} \in \DirSet} \Expt
    [\norm{\ItErr{\Time}}^2], 
\end{align*}
similarly defined for approximation error
\begin{align*} 
\RHOSQMAX{\GoodApp{\dimn}} \; \defn \; \max_{\DirEdge{\fnode}{\snode}
  \in \DirSet} \norm{\GoodApptwo{\dimn}{\LDE}}^2.
\end{align*}
Using the definition of $\AVEERR{\ItErr{\Time}}$, some algebra yields
\begin{align*}
\AVEERR{\ItErr{\Time+1}} - \AVEERR{\ItErr{\Time}} \; & = \; \Exs
\Big[\sum_{j=1}^\dimn \big(\coeffupdown{\Time+1}{j} -
  \coeffupdown{*}{j} \big)^2 \, - \, \sum_{j=1}^\dimn
  \big(\coeffupdown{\Time}{j} - \coeffupdown{*}{j} \big)^2 \Big] \\
\; & = \; \Exs \Big[ \sum_{j=1}^\dimn \big \{ \coeffupdown{\Time+1}{j} -
  \coeffupdown{\Time}{j} \big\} \; \big \{
  \big(\coeffupdown{\Time+1}{j} - \coeffupdown{\Time}{j} \big) \, + \, 2
  \big( \coeffupdown{\Time}{j} - \coeffupdown{*}{j} \big) \big \}
  \Big].
\end{align*}
From the update equation~\eqref{EqnUpdateCoef}, we have
\begin{align*}
\coeffupdown{\Time+1}{j} - \coeffupdown{\Time}{j} \; & = \; \step \:
\big(\coefftilthree{\Time+1}{j} - \coeffupdown{\Time}{j} \big),
\end{align*}
and hence
\begin{align}\label{EqnDecompErrDif}
\AVEERR{\ItErr{\Time+1}} - \AVEERR{\ItErr{\Time}} \; = \;
\NEWTERMONEDOWN + \NEWTERMTWODOWN,
\end{align}
where
\begin{subequations}
\begin{align}
\label{EqnDefnNewtermone}
\NEWTERMONEDOWN \; & \defn \; (\step)^2 \sum_{j=1}^{\dimn} \: \Expt \, \big[
  \big(\coefftilthree{\Time+1}{j} - \coeffupdown{\Time}{j} \big)^2
  \big], \quad \mbox{and} \\
\label{EqnDefnNewtermtwo}
\NEWTERMTWODOWN \; & \defn \; 2 \step \, \sum_{j=1}^{\dimn} \: \Expt\: \Big[
  \big(\coefftilthree{\Time+1}{j} - \coeffupdown{\Time}{j} \big) \,
  \big( \coeffupdown{\Time}{j} - \coeffupdown{*}{j} \big) \Big].
\end{align}
\end{subequations}
The following lemma, proved in Appendix~\ref{AppLemNewterm}, provides
upper bounds on these two terms.

\begin{lemma}
\label{LemNewterm}
For all iterations $\Time = 0, 1, 2, \ldots$, we have
\begin{subequations}
\begin{align}
\label{EqnNewtermoneBound}
\NEWTERMONEDOWN \; & \leq \; 4 (\step)^2 \: \sum_{j=1}^\dimn
\bound_j^2, \quad \mbox{and} \\
\label{EqnNewtermtwoBound}
\NEWTERMTWODOWN \; & \leq \; \step \big ( 1 - \frac{\contfact}{2}
\big) \RHOSQMAX{\GoodApp{\dimn}} \, + \, \step \big ( 1 -
\frac{\contfact}{2} \big) \RHOBARSQMAX{\ItErrNo{\Time}} \, - \, \step
(1 + \frac{\contfact}{2}) \AVEERR{\ItErr{\Time}}.
\end{align}
\end{subequations}
\end{lemma}

We continue upper-bounding $\AVEERR{\ItErr{\Time+1}}$ by substituting
the results of Lemma~\ref{LemNewterm} into
equation~\eqref{EqnDecompErrDif}, thereby obtaining
\begin{align*}
\AVEERR{\ItErr{\Time+1}} \; & \leq \; 4 (\step)^2 \sum_{j=1}^\dimn
\bound_j^2 \, + \, \step \big ( 1 - \frac{\contfact}{2} \big) \:
\RHOSQMAX{\GoodApp{\dimn}} \, + \, \step \big ( 1 -
\frac{\contfact}{2} \big) \: \RHOBARSQMAX{\ItErrNo{\Time}} \, + \,
\Big \{ 1 - \step (1 + \frac{\contfact}{2}) \Big \}
\AVEERR{\ItErr{\Time}} \\
\; & \leq \; 4 (\step)^2 \sum_{j=1}^\dimn \bound_j^2 \, + \, \step
\big (1 - \frac{\contfact}{2} \big) \RHOSQMAX{\GoodApp{\dimn}} \, + \, \big(
1 - \step \contfact \big) \RHOBARSQMAX{\ItErrNo{\Time}}.
\end{align*}
Since this equation holds for all directed edges $\DirEdge{\fnode}{\snode}$,
taking the maximum over the left-hand side yields the recursion
\begin{align}
\label{EqnJay}
\RHOBARSQMAX{\ItErrNo{\Time+1}} \; & \leq \; (\step)^2 \, \SUMBOUND^2
\, + \, \step \big ( 1 - \frac{\contfact}{2} \big) \:
\RHOSQMAX{\GoodApp{\dimn}} \, + \, \big( 1 - \step \contfact \big) \: 
\RHOBARSQMAX{\ItErrNo{\Time}},
\end{align}
where we have introduced the shorthand $\SUMBOUND^2 = 4
\sum_{j=1}^\dimn \bound_j^2$.  Setting $\step = 1/(\contfact \,
  (t+1))$ and unwrapping this recursion, we find that
\begin{align*}
\RHOBARSQMAX{\ItErrNo{\Time+1}} \; & \leq \;
\frac{\SUMBOUND^2}{\contfact^2} \sum_{\newt=1}^{\Time+1}
\frac{1}{\newt \, (\Time+1)} \, + \, \frac{2 - \contfact}{2 \contfact} \,
\RHOSQMAX{\GoodApp{\dimn}} \\
\; & \leq \; \frac{2 \, \SUMBOUND^2}{\contfact^2} \; \frac{
  \log(\Time+1)}{\Time+1} \, + \, \frac{1}{\contfact} \,
\RHOSQMAX{\GoodApp{\dimn}},
\end{align*} 
which establishes the claim.


\subsection{Proof of Theorem~\ref{ThmKernel}}
\label{SecProofThmKernel}

As discussed earlier, each iteration of the \ALG algorithm requires
$\order( \dimn)$ operations per edge.  Consequently, it suffices to
show that running the algorithm with $\dimn = \dimcrit$ coefficients
for $(\sum_{j=1}^{\dimn}\myeig_j^2)(1/\delta) \log(1/\delta)$
iterations suffices to achieve mean-squared error less than $\delta$.

The bound~\eqref{EqnGeneralBound} consists of two terms. In order to
characterize the first term (estimation error), we need to bound
$\bound_j$ defined in \eqref{EqnDefnBound}. Using the orthonormality
of the basis functions and the fact that the supremum is attainable
over the compact space $\Space$, we obtain
\begin{align*}
\bound_j \; = \; \max_{\DirEdge{\fnode}{\snode} \in \DirSet} \sup_{y
  \in \Space} \; \frac{\myeig_j \:
  \basis_j(y)}{\int_{\Space}\epot(x,y) \measure(dx)} \; = \;
\order(\myeig_j).
\end{align*}
Therefore, the estimation error decays at the rate $\order
\big((\sum_{j=1}^{\dimn} \myeig_j^2) \: (\log{\Time}/\Time) \big)$, so
that $\Time = \order\big( (\sum_{j=1}^{\dimn} \myeig_j^2) (1/\delta)
\log(1/\delta) \big)$ iterations are sufficient to reduce it to
$\order(\delta)$.

The second term (approximation error) in the bound
\eqref{EqnGeneralBound} depends only on the choice of $\dimn$, and in
particular on the $\dimn$-term approximation error
$\|\GoodAppDown{\dimn}\|_{L^2}^2 = \| \mes^*_\LDE -
\Pir(\mes^*_\LDE)\|_{L^2}^2$. To bound this term, we begin by
representing $\mes^*_\LDE$ in terms of the basis expansion
$\sum_{j=1}^\infty \coeff^{*}_j \basis_j$. By the Pythagorean theorem,
we have
\begin{align}
\label{EqnPythagError}
\norm{ \mes^*_\LDE - \Pir(\mes^*_\LDE)}^2 & = \sum_{j =
  \dimn+1}^\infty (\coeff^{*}_j)^2.
\end{align}

Our first claim is that $\sum_{j=1}^\infty (\coeff^*_j)^2/\myeig_j <
\infty$. Indeed, since $\mes^*$ is a fixed point of the message
update equation, we have
\begin{align*}
\mes^*_{\LDE}(\cdot) \; \propto \; \int_{\Space} \epot(\cdot,
\realize_\fnode) \, M(\realize_\fnode) \, \measure(d \realize_\fnode),
\end{align*}
where $M(\cdot) \defn \psi_\fnode(\cdot) \prod_{\CompNeig}
\mes^*_{\LowerDirEdge{w}{v}}(\cdot)$. Exchanging the order of
integrations using Fubini's theorem, we obtain
\begin{align}
\label{EqnCoffee}
\coeff^*_j \; = \; \inprod{\mes^*_{\LDE}}{\basis_j} \; & \propto \;
\int_{\Space} \inprod{\basis_j(\cdot)}{\epot(\cdot,
  \realize_\fnode)} \; M(\realize_\fnode) \: \measure(d
\realize_\fnode).
\end{align}
By the eigenexpansion of $\epot$, we have
\begin{align*}
\inprod{\basis_j(\cdot)}{\epot(\cdot, \realize_\fnode)} \; & =
\; \sum_{k=1}^\infty \myeig_k \inprod{\basis_j}{\basis_k} \;
\basis_k(\realize_\fnode) \; = \; \myeig_j \:
\basis_j(\realize_\fnode).
\end{align*}
Substituting back into our initial equation~\eqref{EqnCoffee}, we find
that
\begin{align*}
\coeff^*_j \; & \propto \; \myeig_j \int_{\Space} \basis_j(\realize_\fnode)
\; M(\realize_\fnode) \: \measure(d \realize_\fnode) \; = \; \myeig_j \:
\widetilde{\coeff}_j,
\end{align*}
where $\widetilde{\coeff}_j$ are the basis expansion coefficients of
$M$.  Since the space $\Space$ is compact, one can see that $M
\in L^2(\Space)$, and hence $\sum_{j=1}^\infty
\widetilde{\coeff}_j^2 < \infty$. Therefore, we have 
\begin{align*}
\sum_{j=1}^\infty \frac{(\coeff^*_j)^2}{\myeig_j} \; & \propto \;
\sum_{j=1}^\infty \myeig_j \: \widetilde{\coeff}_j^2 \; < \; +\infty,
\end{align*}
where we used the fact that $\sum_{j=1}^{\infty}\myeig_j < \infty$.

We now use this bound to control the approximation
error~\eqref{EqnPythagError}.  For any $\dimn = 1, 2, \ldots$, we have
\begin{align*}
\sum_{j=\dimn+1}^\infty (\coeff^*_j)^2 \; & = \;
\sum_{j=\dimn+1}^\infty \myeig_j \; \frac{(\coeff^*_j)^2}{\myeig_j} \;
\leq \; \myeig_{\dimn} \sum_{j=\dimn+1}^\infty
\frac{(\coeff^*_j)^2}{\myeig_j} \; = \; \order(\myeig_{\dimn}),
\end{align*}
using the non-increasing nature of the sequence $\{ \myeig_j
\}_{j=1}^\infty$.  Consequently, by definition of $\dimcrit$
\eqref{EqnDefRstar}, we have
\begin{align*}
\| \mes^*_\LDE - \Pi^{\dimcrit}(\mes^*_\LDE) \|_{L^2}^2 & = \order(\delta),
\end{align*}
as claimed.


\section{Conclusion}
\label{SecConclusion}

Belief propagation is a widely used message-passing algorithm for
computing (approximate) marginals in graphical models.  In this paper,
we have presented and analyzed the \ALG algorithm for running BP in
models with continuous variables.  It is based on two forms of
approximation: a \emph{deterministic approximation} that involves
projecting messages onto the span of $\dimn$ basis functions, and a
\emph{stochastic approximation} that involves approximating integrals
by Monte Carlo estimates.  These approximations, while leading to an
algorithm with substantially reduced complexity, are also controlled:
we provide upper bounds on the convergence of the stochastic error,
showing that it goes to zero as $\order(\log \Time / \Time)$
with the number of iterations, and also control on the deterministic
error.  For graphs with relatively smooth potential functions, as
reflected in the decay rate of their basis coefficients, we provided a
quantitative bound on the total number of basic arithmetic operations
required to compute the BP fixed point to within $\delta$-accuracy.
We illustrated our theoretical predictions using experiments on
simulated graphical models, as well as in a real-world instance
of optical flow estimation.

Our work leaves open a number of interesting questions.  First,
although we have focused exclusively on models with pairwise
interactions, it should be possible to develop forms of \ALG for
higher-order factor graphs.  Second, the bulk of our analysis was
performed under a type of contractivity condition, as has been used in
past work~\cite{Tatikonda02,Ihler05,MooijKapp07,RoostaEtal08} on
convergence of the standard BP updates.  However, we suspect that this
condition might be somewhat relaxed, and doing so would demonstrate
applicability of the \ALG algorithm to a larger class of graphical
models.

\subsection*{Acknowledgements}

This work was partially supported by Office of Naval Research MURI
grant N00014-11-1-0688 to MJW.  The authors thank Erik Sudderth for
helpful discussions and pointers to the literature.


\appendix

\section{Proof of Lemma~\ref{LemTwoTerm}}
\label{AppLemTwoTerm}

Subtracting $\coeffbptwo{j}$ from both sides of the
update~\eqref{EqnUpdateCoef} in Step 2(c), we obtain
\begin{align}
\label{EqnRecursionOne}
\coeffdirthree{\Time+1}{j} - \coeffbptwo{j} \; & = \; (1 - \step) \,
\big[ \coeffdirthree{\Time}{j} - \coeffbptwo{j} \big ]\, + \, \step \:
\big[ \upfuncoeftwo{\Time}{j} - \coeffbptwo{j} \big] \, + \, \step \:
\deviatetwo{\Time+1}{j}.
\end{align}
Setting $\step = 1 / (\Time + 1)$ and unwrapping the
recursion~\eqref{EqnRecursionOne} then yields
\begin{align*}
\coeffdirthree{\Time+1}{j} - \coeffbptwo{j} \; & = \;
\frac{1}{\Time+1} \sum_{\newt = 0}^{\Time}
\big[\upfuncoeftwo{\newt}{j} - \coeffbptwo{j} \big] \, + \,
\frac{1}{\Time+1} \sum_{\newt=0}^{\Time} \deviatetwo{\newt +1}{j}.
\end{align*}
Squaring both sides of this equality and using the upper bound
$(a+b)^2 \leq 2 a^2 + 2 b^2$, we obtain
\begin{align*}
\big(\coeffdirthree{\Time+1}{j} - \coeffbptwo{j} \big)^2 \; & \leq \;
\frac{2}{(\Time+1)^2} \: \Big \{ \sum_{\newt=0}^{\Time} \: \big[
  \upfuncoeftwo{\newt}{j} - \coeffbptwo{j} \big] \Big \}^2 \, + \,
\frac{2}{(\Time+1)^2} \: \Big \{ \sum_{\newt=0}^{\Time}
\deviatetwo{\newt +1}{j} \Big \}^2.
\end{align*}
Summing over indices $j = 1, 2, \ldots, \dimn$ and recalling the
expansion~\eqref{EqnPars}, we find that
\begin{align*}
\| \ItErr{\Time}\|_{L^2}^2 \; & \leq \; \sum_{j=1}^\dimn \Biggr \{
\frac{2}{(\Time+1)^2} \: \Big \{ \sum_{\newt=0}^{\Time} \: \big[
  \upfuncoeftwo{\newt}{j} - \coeffbptwo{j} \big] \Big \}^2 \, + \,
\frac{2}{(\Time+1)^2} \: \Big \{ \sum_{\newt=0}^{\Time}
\deviatetwo{\newt +1}{j} \Big \}^2 \Biggr \} \nonumber \\
\; & \stackrel{\mathrm{(i)}}{\leq} \; 
\underbrace{\frac{2}{(\Time+1)} \: \sum_{j=1}^\dimn
  \sum_{\newt=0}^{\Time} \: \big[ \upfuncoeftwo{\newt}{j} -
    \coeffbptwo{j} \big]^2}_{\mbox{Deterministic term $\detterm$}} \,
+ \, \underbrace{\frac{2}{(\Time+1)^2} \: \sum_{j=1}^\dimn \Big \{
  \sum_{\newt=0}^{\Time} \deviatetwo{\newt +1}{j} \Big \}^2.
}_{\mbox{Stochastic term $\stocterm$}}
\end{align*}
Here step (i) follows from the elementary inequality
\begin{align*}
\Big \{ \sum_{\newt=0}^{\Time} \big[ \upfuncoeftwo{\newt}{j} -
\coeffbptwo{j} \big] \Big \}^2 \; & \leq \; (\Time+1) \:
\sum_{\newt=0}^{\Time} \big[ \upfuncoeftwo{\newt}{j} - \coeffbptwo{j}
\big]^2.
\end{align*}


\section{Proof of Lemma~\ref{LemLipschitz}}
\label{AppLemLipschitz}

Recall the probability density
\begin{align*}
\phack{\mes}{\cdot} \; & \propto \; \margfun_{\snode\fnode}(\cdot) \;
\prod_{\CompNeig} \!\!  \mes_{\LowerDirEdge{\tnode}{\fnode}}(\cdot)
\end{align*}
defined in Step 2 of the \ALG algorithm. Using this shorthand
notation, the claim of Lemma~\ref{LemCuteObservation} can be
re-written as $[\upfun_{\LowerDirEdge{\fnode}{\snode}}(\mes)](x) =
\inprod{\compatfun_{\snode\fnode}(x,\cdot)}{\phack{\mes}{\cdot}}$.
Therefore, applying the Cauchy-Schwartz inequality yields
\begin{align*}
|[\upfun_{\LowerDirEdge{\fnode}{\snode}}(\mes)](x) -
[\upfun_{\LowerDirEdge{\fnode}{\snode}}(\mes^{\prime})](x)|^2 \; \le
\; \norm{\compatfun_{\snode \fnode}(x, \cdot)}^2 \; \,
\norm{\phackno{\mes} \, - \, \phackno{\mes^{\prime}}}^2.
\end{align*}
Integrating both sides of the previous inequality over $\Space$ and
taking square roots yields
\begin{align*}
\norm{\upfun_{\LowerDirEdge{\fnode}{\snode}}(\mes) \, - \,
  \upfun_{\LowerDirEdge{\fnode}{\snode}}(\mes^{\prime})} \; \le \;
\const_{\snode\fnode} \: \norm{\phackno{\mes} \, - \,
  \phackno{\mes^{\prime}}},
\end{align*}
where we have denoted the constant $\const_{\snode\fnode} \defn
\big(\int_{\Space} |\compatfun_{\snode\fnode}(x,y)|^2 \measure (dy)
\measure (dx) \big)^{1/2}$.

Next step would be to upper bound the term $\norm{\phackno{\mes} -
  \phackno{\mes^{\prime}}}$.  In order to do so, we first show that
$\phackno{\mes}$ is a Frechet differentiable operator on the
space $\FunSpaceTwo \defn \convhull\{\mes^{\ast},
\oplus_{(\fnode\to\snode) \in \DirSet} \:\FunSpaceTwo_{\LDE}\}$, where
\begin{align*}
\FunSpaceTwo_{\LDE} \; \defn \; \Big\{ \meshat_{\LDE} \: \Big| \:
\meshat_{\LDE} = \Big[\Expt_{\rvy\sim f}\big[\Pir
    \big(\compatfun_{\snode\fnode} (\cdot, \rvy)\big)\big]\Big]_{+},
\; \text{for some probability density $f$} \Big\},
\end{align*}
denotes the space of all feasible \ALG messages on the directed edge
$\DirEdge{\fnode}{\snode}$. Doing some calculus using the chain rule,
we calculate the partial directional (Gateaux) derivative of the
operator $\phackno{\mes}$ with respect to the function
$\mes_{\LowerDirEdge{\tnode}{\fnode}}$. More specifically, for an
arbitrary function $\arbfun$, we have 
\begin{align*}
[\Deriv (\mes)] (\arbfun) \; = & \; \frac{\margfun_{\snode\fnode}
  \prod_{\CompNeigtwo}
  \mes_{\LowerDirEdge{\ftnode}{\fnode}}}{\inprod{\mesprodtwo}{\margfun_{\snode\fnode}}}
\; \arbfun \\ 
\, & - \,
\frac{\margfun_{\snode\fnode}\mesprod_{\snode\fnode}}{\inprod{\mesprod_{\snode\fnode}}{\margfun_{\snode\fnode}}^2}\;
\inprod{\arbfun}{\margfun_{\snode\fnode}\!\!\!\!\prod_{\CompNeigtwo}
  \!\!\!\! \mes_{\LowerDirEdge{\ftnode}{\fnode}}},
\end{align*}
where $\mesprodtwo = \prod_{\CompNeig}
\mes_{\LowerDirEdge{\tnode}{\fnode}}$. Clearly the Gateaux derivative
is linear and continuous. It is also bounded as will be shown
now. Massaging the operator norm's definition, we obtain
\begin{align}\label{EqnOpNorm}
\nonumber \sup_{\mes\in\FunSpaceTwo} \opnorm{\Deriv (\mes)} \; =& \;
\sup_{\mes\in\FunSpaceTwo} \sup_{\arbfun \in
  \FunSpaceTwo_{\LowerDirEdge{\tnode}{\fnode}}} \frac{\norm{[\Deriv
      (\mes)] (\arbfun)}}{\norm{\arbfun}} \\ \nonumber \; \le & \;
\sup_{\mes\in\FunSpaceTwo} \frac{ \sup_{x\in\Space} \;
  \margfun_{\snode\fnode}(x) \prod_{\CompNeigtwo}
  \mes_{\LowerDirEdge{\ftnode}{\fnode}}(x)}{\inprod{\mesprod_{\snode\fnode}}{\margfun_{\snode\fnode}}}
\\ \; & + \; \sup_{\mes\in\FunSpaceTwo}
\frac{\norm{\margfun_{\snode\fnode}\mesprod_{\snode\fnode}} \:
  \norm{\margfun_{\snode\fnode}\prod_{\CompNeigtwo}
    \mes_{\LowerDirEdge{\ftnode}{\fnode}}}}{\inprod{\mesprod_{\snode\fnode}}{\margfun_{\snode\fnode}}^2}.
\end{align}
Since the space $\Space$ is compact, the continuous functions
$\margfun_{\snode\fnode}$ and $\mes_{\LowerDirEdge{\ftnode}{\fnode}}$
achieve their maximum over $\Space$. Therefore, the numerator
of~\eqref{EqnOpNorm} is bounded and we only need to show that the
denominator is bounded away from zero. 

For an arbitrary message $\mes_{\LDE} \in \FunSpaceTwo_{\LDE}$ there
exist $0 < \alpha < 1$ and a bounded probability density $f$ so that
\begin{align*}
\mes_{\LDE}(x) \, = \, \alpha \: \mes_{\LDE}^{\ast}(x)
\: + \: (1 - \alpha) \Big[\Expt_{\rvy\sim f}
   \big[ \prjcompatfun_{\snode\fnode}(x, \rvy)\big]\Big]_+,
\end{align*}
where we have introduced the shorthand
$\prjcompatfun_{\snode\fnode}(\cdot, y) \defn \Pir
(\compatfun_{\snode\fnode}(\cdot, y) )$. According to
Lemma~\ref{LemCuteObservation}, we know $\mes^{\ast}_{\LDE} =
\Expt_{Y}[\compatfun_{\snode\fnode}(\cdot, Y)]$, where
$Y\sim\phackno{\mes^{\ast}}$. Therefore, denoting
$\probdens^{\ast} = \phackno{\mes^{\ast}}$, we have 
\begin{align}\label{EqnMesLowBnd}
\nonumber \mes_{\LDE}(x) \; & \ge \; \alpha \: \Expt_{\rvy\sim
    \probdens^{\ast}} [\compatfun_{\snode\fnode}(x, \rvy)] \: + \: (1
  - \alpha) \: \Expt_{\rvy\sim
    f}[\prjcompatfun_{\snode\fnode}(x, \rvy) ]\\
\; & = \; \Expt_{\rvy\sim (\alpha \probdens^{\ast} +
  (1-\alpha)f)}[\prjcompatfun_{\snode\fnode}(x, \rvy)] \: + \: \alpha
    \: \Expt_{\rvy\sim \probdens^{\ast}} [\compatfun_{\snode\fnode}(x,
      \rvy) - \prjcompatfun_{\snode\fnode}(x, \rvy)].
\end{align}
On the other hand, since $\Space$ is compact, we can
exchange the order of expectation and projection using Fubini's
theorem to obtain
\begin{align*}
\Expt_{\rvy\sim \probdens^{\ast}} [\compatfun_{\snode\fnode}(\cdot,
  \rvy) - \prjcompatfun_{\snode\fnode}(\cdot, \rvy)] \; = \;
\mes^{\ast}_{\LDE} - \Pir(\mes^{\ast}_{\LDE}) \; = \;
\GoodAppDown{\dimn}.
\end{align*}
Substituting the last equality into the bound~\eqref{EqnMesLowBnd}
yields
\begin{align*}
\mes_{\LDE}(x) \; \ge \; \inf_{y \in \Space}
\prjcompatfun_{\snode\fnode}(x, y) \, - \, |\GoodAppDown{\dimn}(x)|.
\end{align*}
Recalling the assumption~\eqref{EqnAnnoy}, one can conclude that the
right hand side of the above inequality is positive for all directed
edges $\DirEdge{\fnode}{\snode}$. Therefore, the denominator of the
expression~\eqref{EqnOpNorm} is bounded away from zero and more
importantly $\sup_{\mes\in\FunSpace} \opnorm{\Deriv (\mes)}$ is
attainable.

Since the derivative is a bounded, linear, and continuous operator,
the Gateaux and Frechet derivatives coincides and we can use
Proposition 2 (Luenberger \cite{Luenberger69}, page 176) to obtain the
following upper bound
\begin{align*}
\norm{ \phackno{\mes} \, - \, \phackno{\mes^{\prime}}} \; \; \le \!\!
\sum_{\CompNeig} \sup_{0\le\alpha\le 1} \opnorm{\Deriv(\mes^{\prime} +
  \alpha \: (\mes - \mes^{\prime}))} \:
\norm{\mes_{\LowerDirEdge{\tnode}{\fnode}} \: - \:
  \mes^{\prime}_{\LowerDirEdge{\tnode}{\fnode}}}.
\end{align*}
Setting $\Lips \: \defn \:
\const_{\snode\fnode} \: \sup_{\mes\in\FunSpaceTwo}
\opnorm{\Deriv (\mes)}$ and putting the pieces together yields 
\begin{align*}
\norm{ \upfun_{\LDE}(\mes) \, - \, \upfun_{\LDE} (\mes^{\prime})} \;
\le & \; \!\!  \sum_{\CompNeig} \! \Lips \:
\norm{\mes_{\LowerDirEdge{\tnode}{\fnode}} \: - \:
  \mes^{\prime}_{\LowerDirEdge{\tnode}{\fnode}}},
\end{align*}
for all $\mes, \mes^{\prime} \in \FunSpaceTwo$.

The last step of the proof is to verify that $\mes^{\ast} \in
\FunSpaceTwo$, and $\meshat^{\Time} \in \FunSpaceTwo$ for all $\Time =
1, 2, \ldots$. By definition we have $\mes^{\ast}\in\FunSpace$. On the
other hand, unwrapping the update~\eqref{EqnUpdateCoef} we obtain
\begin{align*}
\coeff_{\LDE;j}^{\Time} \; = & \; \frac{1}{\Time} \:
\sum_{\newt = 0}^{\Time-1} \: \coefftilthree{\newt+1}{j} \\ 
\; = & \;
\frac{1}{\Time} \: \sum_{\newt = 0}^{\Time-1} \: \frac{1}{\numsampY}
\: \sum_{\ell=1}^{\numsampY} \: \int_{\Space}
\compatfun_{\snode\fnode} (x, \PlainSamp_\ell) \: \basis_j (x) \:
\measure(dx)\\ \; = & \; \int_{\Space}
\Expt_{\rvy\sim\hat{\probdens}}
     [\compatfun_{\snode\fnode} (x, \rvy)] \: \basis_j(x) \:
     \measure(dx),
\end{align*}
where $\hat{\probdens}$ denotes the empirical probability
density. Therefore, $\mes_{\LDE}^{\Time} = \sum_{j=1}^{\dimn}
\mescoef_{\LDE;j}^{\Time} \: \basis_j$ is equal to
$\Pir(\Expt_{\rvy\sim\hat{\probdens}}[\compatfun_{\snode\fnode}(\cdot,
  \rvy)] )$, thereby completing the proof.



\section{Proof of Lemma~\ref{LemMartingale}}
\label{AppLemMartingale}

We begin by taking the conditional expectation of
$\coefftilthree{\Time+1}{j}$, previously
defined~\eqref{EqnUpdateInov}, given the filteration $\Field^\Time$
and with respect to the random samples $\{\PlainSamp_1, \ldots,
\PlainSamp_\numsampY \}\stackrel{\text{i.i.d.}}{\sim}
\phack{\meshat}{\cdot}$. Exchanging the order of expectation and
integral\footnote{Since $\compatfun_{\snode\fnode}(x, y) \basis_i(x)
  \phack{\meshat^{\Time}}{y}$ is absolutely integrable, we can
  exchange the order of the integrals using Fubini's theorem.}  and
exploiting the result of Lemma~\ref{LemCuteObservation}, we obtain
\begin{align}\label{EqnResultFromLemOne}
\Expt [\coefftilthree{\Time+1}{j}\: |
  \:\Field^{\Time}] \; = \; \int_{\Space}
[\upfun_{\LDE}(\meshat^{\Time})](x) \: \basis_{j}(x) \:
\measure(dx) \; = \; \upfuncoeftwo{\Time}{j},
\end{align}
and hence $\Expt [\deviatetwo{\Time +1}{j} \:|\: \Field^{\Time}] = 0$,
for all $j = 1, 2, \ldots, \dimn$ and all directed edges
$\DirEdge{\fnode}{\snode} \in \DirSet$.  Also it is clear that
$\deviatetwo{\Time +1}{j}$ is $\Field^{\Time}$-measurable. Therefore,
$\{\deviatetwo{\newt +1}{j}\}_{\newt = 0}^{\infty}$ forms a martingale
difference sequence with respect to the filtration
$\{\Field^{\newt}\}_{\newt=0}^{\infty}$. 

On the other hand, recalling the bound~\eqref{EqnDefnBound}, we have
\begin{align}\label{EqnBoundStochCoef}
|\coefftilthree{\Time+1}{j}| \; \le \;
\frac{1}{\numsampY} \: \sum_{\ell=1}^{\numsampY}
|\inprod{\compatfun_{\snode\fnode}(\cdot, \PlainSamp_{\ell})}{\basis_j}| \;
\le \; \bound_j.
\end{align}
Moreover, exploiting the result of Lemma~\ref{LemCuteObservation} and
exchanging the order of the integration and expectation once more
yields
\begin{align}\label{EqnBndonCoef}
|\upfuncoeftwo{\Time}{j}| \; = \;
|\inprod{\Expt_{Y}[\compatfun_{\snode\fnode}(\cdot,Y)]}{\basis_j}| \;
= \; |\Expt_{Y} [\inprod{\compatfun_{\snode\fnode}(\cdot, Y)}{\basis_j}]| \; \le \; \bound_j,
\end{align}
where we have $\PlainSamp\sim\phack{\meshat^{\Time}}{y}$.
Therefore, the martingale difference sequence is bounded, in
particular with
\begin{align*}
|\deviatetwo{\Time+1}{j}| \; \le & \; |\coefftilthree{\Time+1}{j}| \, +
|\, \upfuncoeftwo{\Time}{j}| \; \le \; 2\: \bound_j.
\end{align*}


\section{Proof of Lemma~\ref{LemUpperExp}}
\label{AppLemUpperExp}

We start by uniformly upper-bounding the terms
$\Expt[|\term^{\Time+1}_{\LDE}|]$.  To do so we first need to bound
$\norm{\ItErrNo{\Time}_{\LDE}}$. By definition we know
$\norm{\ItErrNo{\Time}_{\LDE}}^2 = \sum_{j=1}^{\dimn}
[\coeff^{\Time}_{\LDE;j} - \coeff^{\ast}_{\LDE;j}]^2$; therefore we
only need to control the terms $\coeff^{\Time}_{\LDE;j}$ and $
\coeff^{\ast}_{\LDE;j}$ for $j=1, 2, \ldots, \dimn$.

By construction, we always have $|\coefftilthree{\Time+1}{j}| \le
\bound_j$ for all iterations $\Time = 0, 1, \ldots$. Also, assuming
that $|\coeff^{0}_{\LDE;j}| \le \bound_j$, without loss of generality,
a simple induction using the update equation~\eqref{EqnUpdateCoef}
shows that $|\coeff^{\Time}_{\LDE;j}| \le \bound_j$ for all
$\Time$. Moreover, using a similar argument leading
to~\eqref{EqnBndonCoef}, we obtain
\begin{align*}
|\coeff^{\ast}_{\LDE;j}| \; = \;
|\inprod{\Expt_{Y}[\compatfun_{\snode\fnode}(\cdot,Y)]}{\basis_j}| \;
= \; |\Expt_{Y} [\inprod{\compatfun_{\snode\fnode}(\cdot,
Y)}{\basis_j}]| \; \le \; \bound_j,
\end{align*}
where we have $\PlainSamp\sim\phack{\mes^{\ast}}{y}$. Therefore,
putting the pieces together, recalling the
definition~\eqref{EqnSmallTerm} of $\term^{\Time+1}_{\LDE}$ yields
\begin{align*}
\Expt[|\term_{\LDE}^{\Time+1}|] \; \le \;
\frac{2}{\Time+1} \: \sum_{\CompNeig}\!\! \matentry \:
\sum_{j=1}^{\dimn} \bound_j^2 \, + \, \frac{32}{\Time+1} \:
\sum_{j=1}^{\dimn}\bound_j^2.
\end{align*}
Concatenating the previous scalar inequalities yields $\Expt
[\term_{0}^{\Time+1}] \coneleq \vvec/(\Time+1)$, for all
$\Time \ge0$, where we have defined the $\Dimn$-vector $\vvec \defn
\big \{ \sum_{j=1}^{\dimn} \bound_j^2 \big \} (2 \NilMat \onevec +
32)$.

We now show, using an inductive argument, that
\begin{align}
\label{EqnBoundTterm}
\Expt [\term_{s}^{\Time+1}] \; \coneleq \; \frac{\vvec}{ \Time+1}\:
\sum_{u = 0}^{s} \frac{(\log(\Time+1))^{u}}{u!},
\end{align}
for all $s = 0, 1, 2, \ldots$ and $\Time = 0, 1, 2, \ldots$.  We have
already established the base case $s = 0$.  For some $s > 0$, assume
that the claim holds for $s-1$. By the definition of
$\term_{s}^{\Time+1}$, we have
\begin{align*}
\Expt [\term_{s}^{\Time+1}] \; = & \; \frac{1}{\Time+1} \:
\sum_{\newt=1}^{\Time} \Expt [\term_{s-1}^{\newt}] \\
\; \coneleq & \; \frac{\vvec}{\Time+1} \: \sum_{\newt=1}^{\Time} \Big \{
\frac{1}{\newt} \: + \: \sum_{u=1}^{s-1}
\frac{(\log\newt)^{u}}{u! \: \newt} \Big \},
\end{align*}
where the inequality follows from the induction hypothesis.  We now
make note of the elementary inequalities $\sum_{\newt=1}^{\Time}
1/\newt \le 1 + \log\Time$, and
\begin{align*}
\sum_{\newt=1}^{\Time} \frac{(\log\newt)^u}{u! \: \newt} \; \le \;
\int_{1}^{\Time} \frac{(\log x)^u}{u! \: x} \: dx \; = \;
\frac{(\log\Time)^{(u+1)}}{(u+1)!}, \qquad \mbox{for all $u \geq 1$}
\end{align*}
from which the claim follows.


\section{Proof of Lemma~\ref{LemNewterm}}
\label{AppLemNewterm}

\paragraph{Upper-bounding the term $\NEWTERMONEDOWN$:}


By construction, we always have $|\coefftilthree{\Time+1}{j}| \leq
\bound_j$ for all iterations $\Time = 0, 1, 2, \ldots$.  Moreover,
assuming $|\coeffupdown{0}{j}| \leq \bound_j$, without loss of
generality, a simple induction on the update equation shows that
$|\coeffupdown{\Time}{j}| \leq \bound_j$ for all iterations $\Time =
0, 1, \ldots$.  On this basis, we find that
\begin{align*}
\NEWTERMONEDOWN \; & = \; (\step)^2 \: \sum_{j=1}^{\dimn} \: \Expt \,
\big[ \big(\coefftilthree{\Time+1}{j} - \coeffupdown{\Time}{j} \big)^2 \big]
\; \leq \; 4 (\step)^2 \, \sum_{j=1}^\dimn \bound_j^2,
\end{align*}
which establishes the bound~\eqref{EqnNewtermoneBound}.

\paragraph{Upper-bounding the term $\NEWTERMTWODOWN$:}
It remains to establish the bound~\eqref{EqnNewtermtwoBound} on
$\NEWTERMTWODOWN$.  We first condition on the $\sigma$-field
$\Field^{\Time} = \sigma(\mes^0, \ldots, \mes^\Time)$ and take
expectations over the remaining randomness, thereby obtaining
\begin{align*}
\NEWTERMTWODOWN \; & = \; 2 \step \, \Expt \Big[ \Expt \big[
    \sum_{j=1}^{\dimn} \big(\coefftilthree{\Time+1}{j} -
    \coeffupdown{\Time}{j} \big) \, \big( \coeffupdown{\Time}{j} -
    \coeffupdown{*}{j} \big) \, \big| \, \Field^{\Time} \big] \Big]\\ 
\; & = \; 2 \step \, \Expt \Big[ \sum_{j=1}^{\dimn} \big(
  \upfuncoeftwo{\Time}{j} - \coeffupdown{\Time}{j} \big) \, \big(
  \coeffupdown{\Time}{j} - \coeffupdown{*}{j} \big) \Big],
\end{align*}
where $\{\upfuncoeftwo{\Time}{j}\}_{j=1}^{\infty}$ are the expansion
coefficients of the function $\upfun_{\LDE}(\meshat^t)$
(i.e. $\upfuncoeftwo{\Time}{j} =
\inprod{\upfun_{\LDE}(\meshat^t)}{\basis_j}$), and we have recalled
the result $\Expt[\coefftilthree{\Time+1}{j} | \Field^{\Time}] =
\upfuncoeftwo{\Time}{j}$ from~\eqref{EqnResultFromLemOne}. By
Parseval's identity, we have
\begin{align*}
T \; & \defn \; \sum_{j=1}^{\dimn} \big(\upfuncoeftwo{\Time}{j} -
\coeffupdown{\Time}{j} \big) \, \big( \coeffupdown{\Time}{j} -
\coeffupdown{*}{j} \big) \\
\; & = \; \inprod{\Pir(\upfun_{\LowerDirEdge{\fnode}{\snode}}(\meshat^t)) -
  \mes^t_{\LDE}}{\mes^t_{\LDE} - \Pir(\mes^*_{\LDE})}.
\end{align*}
Here we have used the basis expansions
\begin{align*}
\mes^t_{\LowerDirEdge{\fnode}{\snode}} = \sum_{j=1}^\dimn
\coeffupdown{\Time}{j} \basis_j, \quad \mbox{and} \quad
\Pir(\mes^*_{\LowerDirEdge{\fnode}{\snode}}) = \sum_{j=1}^\dimn
\coeffupdown{*}{j} \basis_j.
\end{align*}
Since $\Pir(\mes^\Time_\LDE) = \mes^\Time_\LDE$ and
$\upfun_\LDE(\mes^*) = \mes^*_\LDE$, we have
\begin{align*}
T \; & = \; \inprod{\Pir \big(\upfun_\LDE(\meshat^\Time) -
  \upfun_\LDE(\mes^*) \big)}{\mes^t_\LDE - \Pir(\mes^*_\LDE)} \, - \,
\norm{\mes^\Time_\LDE - \Pir(\mes^*_\LDE)}^2 \\
\; & \stackrel{(i)}{\leq} \; \norm{ \Pir \big(\upfun_\LDE(\meshat^\Time) -
\upfun_\LDE(\mes^*) \big)} \; \: \norm{\mes^t_\LDE -
\Pir(\mes^*_\LDE)} \, - \, \norm{\mes^\Time_\LDE -
\Pir(\mes^*_\LDE)}^2 \\
\; & \stackrel{(ii)}{\leq} \; \norm{\upfun_\LDE(\meshat^\Time) -
\upfun_\LDE(\mes^*)} \: \norm{\mes^t_\LDE - \Pir(\mes^*_\LDE)}
\, - \, \norm{\mes^\Time_\LDE - \Pir(\mes^*_\LDE)}^2.
\end{align*}
where step (i) uses the Cauchy-Schwarz inequality, and step (ii) uses
the non-expansivity of projection.
Applying the contraction condition~\eqref{EqnContraction}, we obtain
\begin{align*}
T \; \leq & \; \big ( 1 - \frac{\contfact}{2} \big) \; \sqrt{ \frac{\sum
    \limits_{\CompNeig}
    \norm{\meshat^t_{\LowerDirEdge{w}{v}} -
    \mes^*_{\LowerDirEdge{w}{v}}}^2}{|\neigh(v)| - 1}} \;
\norm{\mes^t_\LDE - \Pir(\mes^*_\LDE)} \\
\, & - \, \norm{\mes^\Time_\LDE -
\Pir(\mes^*_\LDE)}^2 \\
\; \leq & \; \big ( 1 - \frac{\contfact}{2} \big) \biggr \{ \frac{1}{2} \:
\frac{\sum_{\CompNeig}
  \norm{\mes^t_{\LowerDirEdge{w}{v}} -
  \mes^*_{\LowerDirEdge{w}{v}}}^2}{|\neigh(v)| - 1} \, + \, 
\frac{1}{2} \: \norm{\mes^t_\LDE - \Pir(\mes^*_\LDE)}^2 \biggr \} \\
\, & -  \, \norm{\mes^\Time_\LDE - \Pir(\mes^*_\LDE)}^2,
\end{align*}
where the second step follows from the elementary inequality $a b \leq
a^2/2 + b^2/2$ and the non-expansivity of projection onto the space of
non-negative functions.  By the Pythagorean theorem, we have
\begin{align*}
\norm{\mes^t_{\LowerDirEdge{w}{v}} - \mes^*_{\LowerDirEdge{w}{v}}}^2
\; & = \; \norm{\mes^t_{\LowerDirEdge{w}{v}} -
\Pir(\mes^*_{\LowerDirEdge{w}{v}})}^2 \, + \,
\norm{\Pir(\mes^*_{\LowerDirEdge{w}{v}}) -
\mes^*_{\LowerDirEdge{w}{v}}}^2 \\
\; & = \; \norm{\ItErrNo{\Time}_{\LowerDirEdge{\tnode}{\fnode}}}^2 +
\norm{\GoodApptwo{\dimn}{\LowerDirEdge{\tnode}{\fnode}}}^2.
\end{align*}
Using this equality and taking expectations, we obtain
\begin{align*}
\Exs[T] \; & \leq \; \big ( 1 - \frac{\contfact}{2} \big) \biggr \{
\frac{1}{2} \: \frac{\sum_{\CompNeig}
  [\AVEERR{\Delta^\Time_{\LowerDirEdge{w}{v}}} +
  \norm{\GoodApp{\dimn}_{\LowerDirEdge{w}{v}}}^2] } {|\neigh(v)| - 1}
+ \frac{1}{2} \: \AVEERR{\ItErr{\Time}} \biggr \} \, - \, \AVEERR{\ItErr{\Time}}
\\
\; & \leq \; \big ( \frac{1}{2} - \frac{\contfact}{4} \big) \:
\RHOSQMAX{\GoodApp{\dimn}} \, + \, \big ( \frac{1}{2} - \frac{\contfact}{4}
\big) \: \RHOBARSQMAX{\ItErrNo{\Time}} \, - \, (\frac{1}{2} +
\frac{\contfact}{4}) \: \AVEERR{\ItErr{\Time}}.
\end{align*}
Since $\NEWTERMTWODOWN = 2 \step \,\Exs[T]$, the claim follows.


\bibliographystyle{plain}
\bibliography{NewSBP_bibfile}

\end{document}